\documentclass[journal,12pt,onecolumn,draftclsnofoot,]{IEEEtran}

\usepackage[T1]{fontenc}
\usepackage{subcaption}
\usepackage{color}
\usepackage{varioref}
\usepackage{textcomp}
\usepackage{amsthm}
\usepackage{amsmath}
\usepackage{amssymb}
\usepackage{graphicx}
\usepackage{setspace}
\usepackage{esint}
\usepackage{algorithm}
\usepackage{algpseudocode}
\usepackage{pifont}
\usepackage{wrapfig}
\usepackage{multirow}
\usepackage{ctable} % for \specialrule command
\usepackage{algpseudocode}
\usepackage{dsfont}
\usepackage{xcolor}
\usepackage{tabu}
\usepackage{amsmath}

\usepackage{enumitem}
\usepackage{cite}
\usepackage{cleveref}
\usepackage{amsmath}
\usepackage{amssymb}
\usepackage{graphicx,scalerel}
\makeatletter

\newcommand{\Rmnum}[1]{\expandafter\@slowromancap\romannumeral #1@}
\makeatother

\newcommand\wh[1]{\hstretch{2}{\hat{\hstretch{.5}{#1}}}}
\newtheorem{theorem}{Theorem}
\newtheorem{lemma}{Lemma}
\newtheorem{remark}{Remark}
\newtheorem{assumption}{Assumption}

\newtheorem{corollary}{Corollary}
\theoremstyle{definition}

\providecommand{\propositionname}{Proposition}

\ifCLASSINFOpdf

\else

\fi

\usepackage[T1]{fontenc}
\usepackage{etoolbox}

\makeatletter
\patchcmd{\maketitle}{\@fnsymbol}{\@alph}{}{}  % Footnote numbers from symbols to small letters
\makeatother

\title{Blind Federated Edge Learning}
\author{\IEEEauthorblockN{Mohammad Mohammadi Amiri, Tolga M. Duman, Deniz G\"und\"uz,\\ Sanjeev R. Kulkarni, H. Vincent Poor\thanks{M. Mohammadi Amiri, S. R. Kulkarni, and H. V. Poor are with the Department of Electrical Engineering, Princeton University, Princeton, NJ 08544, USA (e-mail: \{mamiri, kulkarni, poor\}@princeton.edu).} \thanks{T. M. Duman is with the Department of Electrical and Electronics Engineering, Bilkent University, Ankara 06800, Turkey (e-mail: duman@ee.bilkent.edu.tr).} \thanks{D. G\"und\"uz is with the Department of Electrical and Electronic Engineering, Imperial College London, London SW7 2AZ, U.K. (e-mail: d.gunduz@imperial.ac.uk).}
}
}
\date{}
\begin{document}
 
\maketitle

%\hspace{-0.5in}
%\thispagestyle{empty}
\begin{abstract}
We study federated edge learning (FEEL), where wireless edge devices, each with its own dataset, learn a global model collaboratively with the help of a wireless access point acting as the parameter server (PS). \makeatletter{\renewcommand*{\@makefnmark}{}
\footnotetext{Part of this work was presented at the IEEE Global Conference on Signal and Information Processing (GlobalSIP), Ottawa, ON, Canada, Nov. 2019 \cite{MohammadTolgaDenizFLGlobalSIP}.}\makeatother}
At each iteration, wireless devices perform local updates using their local data and the most recent global model received from the PS, and send their local updates to the PS over a wireless fading multiple access channel (MAC). 
The PS then updates the global model according to the signal received over the wireless MAC, and shares it with the devices.
%We assume that the transmission from the devices to the PS takes place without transmit side channel state information (CSI), while the PS has imperfect CSI. 
Motivated by the additive nature of the wireless MAC, we propose an analog `over-the-air' aggregation scheme, in which the devices transmit their local updates in an uncoded fashion.
Unlike recent literature on over-the-air edge learning, here we assume that the devices do not have channel state information (CSI), while the PS has imperfect CSI. 
Instead, the PS is equipped multiple antennas to alleviate the destructive effect of the channel, exacerbated due to the lack of perfect CSI. 
We design a receive beamforming scheme at the PS, and show that it can compensate for the lack of perfect CSI when the PS has a sufficient number of antennas.
%, in which case the PS receives the average of the local updates sent from the devices. 
We also derive the convergence rate of the proposed algorithm highlighting the impact of the lack of perfect CSI, as well as the number of PS antennas.
Both the experimental results and the convergence analysis illustrate the performance improvement of the proposed algorithm with the number of PS antennas, where the wireless fading MAC becomes deterministic despite the lack of perfect CSI when the PS has a sufficiently large number of antennas. 
\end{abstract}
%  The PS updates the model parameter after receiving the gradient estimates from the workers, and shares the updated model with the workers in a lossless fashion for the computations of the next DSGD iteration.

%\begin{IEEEkeywords}
%Gaussian broadcast channels, decentralized caching, superposition coding.
%\end{IEEEkeywords}

%\newpage
\section{Introduction}\label{SecIntro}

With the growing prevalence of Internet of things (IoT) devices, constantly collecting information about various physical phenomena, and the growth in the number and processing capabilities of mobile edge devices (phones, tablets, smart watches and activity monitors), there is a growing interest in enabling machine learning (ML) to learn from data distributed across edge devices.
Centralized ML techniques are often developed, assuming that the datasets are offloaded to a central processor. 
In the case of wireless edge devices, centralized ML techniques are not desirable, since offloading such massive amounts of data to a central cloud may be too costly in terms of energy and bandwidth, and may compromise data privacy. 
\textit{Federated learning} (FL) has been developed to enable ML at the wireless edge by pushing the network intelligence to the edge by utilizing the processing capabilities of wireless devices.

With FL, wireless devices train a global model collaboratively using their local datasets, which remain localized enhancing data privacy, with the help of a parameter server (PS) that keeps track of the model \cite{DCKonecnyFederated}.
At each iteration of FL, the PS shares the current global model with the devices, and collects the local model updates from the devices to update the global model. 
This procedure continues until the global model converges, or the devices stop participating in the training because of hitting their limited power budget, or moving out of the coverage of the PS.

FL involves communications over unreliable wireless networks with limited resources, particularly in the device-to-PS direction where a large number of devices, each with limited bandwidth and power, communicate with the PS over a shared wireless medium.  
Therefore, it is vital to design communication-efficient protocols for the realization of an FL framework.
Several approaches have been proposed in recent years to limit the communication requirements in the FL setting \cite{DCKonecnyFederated,McMahan2017CommunicationEfficientLO,GoogleMcMahanFed,KonecnyFLBeyondData,FLWithnonIIDZhao,XLiFedAveFLnonIID,MDSV_NIPS2020}. 
However, these works ignore the physical characteristics of the underlying communication channels for wireless edge learning and consider interference-and-error-free rate-limited communication links.

%While there is already a significant body of research in distributed ML focusing on reducing the communication cost \cite{DCLimitedPrecisionGupta, DCOneBitQuan,DCAlistarhQSGD,DCWenTernGrad, DorefaZhou,DCWangATOMO,SignSGDBernstein,ErrorCompSGDWu, HardwareLi,ScalableDNNStorm,DCAjiSparse,DCSattlerSparseBinary,DCLinHanDeepGradComp,MePropSun,SparCMLRenggli,SparseConvergenceAlistarh,VarBasedTsuzuku,McMahan2017CommunicationEfficientLO,LAGGradientGiannakis,DCKonecnyFederated,DCLinHanDeepGradComp,LocalSGDStich,UseLocalSGDLin},

Recently there have been significant efforts to incorporate physical layer characteristics of wireless networks into FL system design \cite{MohammadDenizDSGDCS,KaibinParallelWork,YangFedLearOverAirComp,MohamamdDenizFLOverAirSPAWC19,FLTWCMohammadDenizFading,MohammadTolgaDenizFLGlobalSIP,TungVuFLMassiveMIMO,YoSebMohammadMIMOFL,CohenAnalogGDDL,RaviCommEffFLGaussian,DenizOneBitAgg,OsvaldoFLPrivacyFree,YangArafaVinceAgeBasedFL,ShiZhouFastConverFL,HowardVinceSchedulingsFL,YuxuanFLRedundantData,MohammadDenizSanjVinceISIT20,JinHyunAhnFLNoisyDownlink,AccelDNNFLEdge,ZengFLCPU_GPU,ChenVinceJointLearCommFL,DinhFLWireNetConver,FLConvergenceMohDenSanjVince,MohammadFLConvNoisyDownlink,ShiFLLatecyConsWirelessFL,DenizCommunLearnEdge,GuoFLCustomizedDesign,FreyOverTheAirCompDisMachineLear,SeryOverAirFLCohenHeterData}, referred to as \textit{federated edge learning} (FEEL).
Several studies have incorporated over-the-air computation into FEEL utilizing the superposition property of the wireless multiple access channel (MAC) for reliable transmission from the devices to the PS, where the MAC naturally provides the sum of the updates from the devices to the PS \cite{MohammadDenizDSGDCS,KaibinParallelWork,YangFedLearOverAirComp,MohamamdDenizFLOverAirSPAWC19,FLTWCMohammadDenizFading,MohammadTolgaDenizFLGlobalSIP,CohenAnalogGDDL,RaviCommEffFLGaussian,DenizOneBitAgg,SeryOverAirFLCohenHeterData,GuoFLCustomizedDesign,FreyOverTheAirCompDisMachineLear}.
Various device scheduling techniques for FEEL have been introduced in order to select a subset of devices sharing the limited wireless resources in each communication round \cite{YangArafaVinceAgeBasedFL,MohammadDenizSanjVinceISIT20,ShiZhouFastConverFL,YuxuanFLRedundantData,HowardVinceSchedulingsFL}.
Also, allocating resources to optimize a performance measure is another active research direction in FEEL \cite{TungVuFLMassiveMIMO,AccelDNNFLEdge,ZengFLCPU_GPU,ChenVinceJointLearCommFL,DinhFLWireNetConver,ShiFLLatecyConsWirelessFL}.   
Several studies have provided convergence guarantees of FEEL under different practical constraints and types of heterogeneity in a federated setting \cite{XLiFedAveFLnonIID,MDSV_NIPS2020,DinhFLWireNetConver,FLConvergenceMohDenSanjVince,MohammadFLConvNoisyDownlink,GuoFLCustomizedDesign}. 
Furthermore, beamforming techniques at the PS with multiple antennas have been designed to improve the quality of the estimated signal used for updating the global model \cite{YangFedLearOverAirComp,MohammadTolgaDenizFLGlobalSIP,YoSebMohammadMIMOFL}.
In \cite{YangFedLearOverAirComp}, a beamforming technique is used at the PS to maximize the number of devices participating in each communication round of training, while \cite{YoSebMohammadMIMOFL} introduces a nonlinear estimation method to recover the sum of updates sent from the devices using their sparsity.

In this paper, we extend our previous work in \cite{MohammadTolgaDenizFLGlobalSIP} and study FEEL over a wireless fading MAC from the devices to the PS.
In order to benefit from the over-the-air computation, we consider uncoded transmission of local model updates from the devices to the PS, whose advantages over digital transmission have been shown in \cite{MohammadDenizDSGDCS,KaibinParallelWork, YangFedLearOverAirComp, MohamamdDenizFLOverAirSPAWC19,FLTWCMohammadDenizFading}. 
Over-the-air computation over a wireless fading MAC requires each transmitting device to  scale its transmission depending on the instantaneous channel state so that they arrive at the same power level at the PS.
This, in turn, necessitates perfect channel state information (CSI) at the devices, acquisition of which would introduce additional delays and reduce the spectral efficiency.
Alternatively, in this work, we consider FEEL with no CSI at the devices and imperfect CSI at the PS.
To the best of our knowledge, this is the first paper in the FEEL literature to consider no CSI at the transmitters (CSIT) and imperfect CSI at the receiver for the device-to-PS transmission.
%We focus on an analog DSGD approach, which exploits the superposition property of the underlying wireless MAC \cite{MohammadDenizDSGDCS}. 
We employ multiple antennas at the PS and design a receive beamformer to overcome the exacerbated negative impact of the underlying wireless fading MAC due to the lack of CSIT and perfect CSI at the PS. 
We analytically show that the proposed beamforming technique alleviates the destructive effects of the interference and noise terms at the PS thanks to the utilization of multiple antennas;
and, in the limit, the fading MAC boils down to a deterministic channel with identical gains from all the devices, which is due to channel hardening \cite{GunnarssonMassiveMIMOChannelHarden}. 
We also provide a convergence analysis of the proposed algorithm, and study the impact of the lack of CSIT and perfect CSI at the PS on the convergence rate.
The convergence analysis shows how the increasing number of antennas at the PS remedies the lack of perfect CSI in the system.  
Numerical experiments on MNIST and CIFAR-10 datasets, corroborated by the analytical convergence results, illustrate the success of the proposed algorithm in combating the unavailability of perfect CSI in the system.
It is shown that, despite the lack of CSI at the devices and perfect CSI at the PS, with sufficiently large number of PS antennas, the proposed algorithm can perform as well as having error-free communication links from the devices to the PS.
%It is worth noting that the CSI requirements of over-the-air computation with a multi-antenna receiver was also studied in \cite{ComOverAirNoCSIT}. The authors proposed a scheme that encodes the information on the energy of the transmitter signals, and hence, limited only to positive values, but requires CSI neither at the transmitters nor at the PS. Performance of this no-CSI scheme for DSGD will be studied in the extended version of this paper. 

%employing massive number of antennas at the receiver to mitigate the fading effect for transmission over a wireless MAC with no CSIT has been studied in the context of computation over-the-air \cite{ComOverAirNoCSIT}. However, the difference is that with computation over-the-air, the receiver is interested in the energy of the received signal to extract the desired signal, whereas in our model the receiver recovers the desired signal directly from received signal. Accordingly, different techniques are developed at the receiver to combat against the fading and estimate the desired signal.  

%\subsection{Notations}\label{SecNot}
\textit{Notations}: $\mathbb{R}$ and $\mathbb{C}$ represent the sets of real and complex values, respectively. 
We denote entry-wise complex conjugate of vector $\boldsymbol{x}$ by $\left( \boldsymbol{x} \right)^*$, and ${\rm{Re}} \{ \boldsymbol{x} \}$ and ${\rm{Im}} \{ \boldsymbol{x} \}$ return entry-wise real and imaginary components of $\boldsymbol{x}$, respectively. 
For $\boldsymbol{x}$ and $\boldsymbol{y}$ with the same dimension, $\boldsymbol{x} \circ \boldsymbol{y}$ returns their element-wise product. 
We denote a zero-mean normal distribution with variance $\sigma^2$ by $\mathcal{N} \left( 0,\sigma^2 \right)$, and $\mathcal{C N} \left( 0,\sigma^2 \right)$ represents a circularly symmetric complex normal distribution with real and imaginary terms each distributed according to $\mathcal{N} \left( 0,\sigma^2 / 2 \right)$. 
We let $[i] \triangleq \{ 1, \dots, i \}$. 
Notation $\left| \cdot \right|$ returns the cardinality of a set or the absolute value of a real number, and the $l_2$ norm of vector $\boldsymbol{x}$ is denoted by $\left\| \boldsymbol{x} \right\|_2$.

\section{System Model}\label{SecProbFormul}

In FL the goal is to minimize a loss function, $F \left( \boldsymbol{\theta} \right)$, where $\boldsymbol{\theta} \in \mathbb{R}^d$ represents the model parameters to be optimized, collaboratively across $M$ devices.
We denote device $m$'s local dataset by $\mathcal{B}_m$ with $B_m \triangleq \left| \mathcal{B}_m \right|$, for $m \in [M]$, and $B \triangleq \sum\nolimits_{m=1}^{M} B_m$. 
We have
\begin{align}\label{GenEmpLossFunc}
F \left( \boldsymbol{\theta} \right) =  \sum\limits_{m=1}^{M} \frac{B_m}{B} F_m \left( \boldsymbol{\theta} \right),   
\end{align}
where $F_m \left( \boldsymbol{\theta} \right)$ represents the average empirical loss at device $m$ with respect to model parameters $\boldsymbol{\theta}$; 
that is,
\begin{align}
F_m \left( \boldsymbol{\theta} \right) = \frac{1}{B_m} \sum\limits_{\boldsymbol{u} \in \mathcal{B}_m} f \left(\boldsymbol{\theta}, \boldsymbol{u} \right), \quad m \in [M],   
\end{align}
where $f \left(\boldsymbol{\theta}, \boldsymbol{u} \right)$ denotes the empirical loss function at data sample $\boldsymbol{u}$ with respect to the model parameters $\boldsymbol{\theta}$ and is defined by the learning task.
Devices perform stochastic gradient descent (SGD) to minimize the loss function $F_m \left( \boldsymbol{\theta} \right)$. 
During global iteration $t$, having received the model parameters $\boldsymbol{\theta} (t)$ from the PS, device $m$ performs $\tau$ local iterations of SGD, with the following update during the $i$-th local iteration:
\begin{align}
\boldsymbol{\theta}_m^{i+1} (t) = \boldsymbol{\theta}_m^i (t) - \eta^i_m (t) \nabla F_m \left( \boldsymbol{\theta}_m^i (t), \xi_m^i (t) \right),  \quad \mbox{$i \in [\tau]$},   
\end{align}
where $\boldsymbol{\theta}_m^1 (t) = {\boldsymbol{\theta}} (t)$, $\eta^i_m (t)$ represents the learning rate, and $\nabla F_m \left( \boldsymbol{\theta}_m^i (t), \xi_m^i (t)  \right)$ denotes the stochastic gradient estimate with respect to $\boldsymbol{\theta}_m^i (t)$ and the local mini-batch sample $\xi_m^i (t)$, chosen uniformly at random from the local dataset $\mathcal{B}_m$, for $m \in [M]$.
We highlight that $\nabla F_m \left( \boldsymbol{\theta}_m^i (t), \xi_m^i (t) \right)$ provides an unbiased estimate of the actual gradient $\nabla F_m \left( \boldsymbol{\theta}_m^i (t) \right)$ with respect to the randomness of the stochastic gradient function; 
that is,
\begin{align}\label{AverageStochGradientEst}
\mathbb{E}_{\xi} \left[ \nabla F_m \left( \boldsymbol{\theta}_m^i (t), \xi_m^i (t) \right) \right] = \nabla F_m \left( \boldsymbol{\theta}_m^i (t)  \right), \quad \forall i \in [\tau], \forall m \in [M], \forall t.   
\end{align}
After performing $\tau$ local updates, device $m$ aims to send its local model update $\Delta \boldsymbol{\theta}_m (t) = \boldsymbol{\theta}_m^{\tau+1} (t) - {\boldsymbol{\theta}} (t)$ to the PS, for $m \in [M]$. 
In the ideal case of receiving the accurate local model updates from the devices, the PS updates the global model according to
\begin{align}\label{ParallelSGDModelUpdate}
\boldsymbol{\theta} (t+1) =\boldsymbol{\theta} (t) + \Delta \boldsymbol{\theta} (t),   
\end{align}
where we have defined
\begin{align}\label{DefDeltaThetat}
\Delta \boldsymbol{\theta} (t) \triangleq    \frac{1}{M} \sum\limits_{m=1}^{M} \Delta \boldsymbol{\theta}_m (t).
\end{align}
However, in our model, the devices transmit their local model updates over a wireless shared medium, which provides the PS with a noisy estimate of $\Delta \boldsymbol{\theta} (t)$. 
In the following, we describe the wireless channels from the devices to the PS, which is equipped with $K$ antennas.

%We consider $M$ devices, where device $m$ has access to a local dataset $\mathcal{B}_m$, and employs SGD to compute the gradient estimate $\boldsymbol{g}_m \left( \boldsymbol{\theta}_t \right) \in \mathbb{R}^d$ at iteration $t$, $m \in [M]$. These local gradient estimates are transmitted to the PS, equipped with $K$ antennas, through a wireless shared medium. The PS updates the model parameter based on its received signal, and shares it with all the devices over an error-free shared link, so that all the devices have a globally consistent model parameter.

We model the shared wireless channel from the devices to the PS with $K$ antennas as a wireless fading MAC, where OFDM is used to divide the available bandwidth into $s$ subchannels, $s \le d$ (in practice, we typically have $s \ll d$). 
We assume that $N$ OFDM symbols can be transmitted over each subchannel at each global iteration. 
The received vector corresponding to the $n$-th OFDM symbol during global iteration $t$ at the $k$-th antenna of the PS is given by
\begin{align}\label{ReceivedVectorPSGenAntennak}
\boldsymbol{y}^n_k (t) = \sum\limits_{m = 1}^{M} \boldsymbol{h}^n_{m,k} (t) \circ \boldsymbol{x}^n_{m} (t) + \boldsymbol{z}^n_k (t), \quad \mbox{$k \in [K]$},
\end{align}
where $\boldsymbol{x}^n_{m} (t)$ is the $n$-th symbol of dimension $s$ transmitted by device $m$, $\boldsymbol{h}^n_{m,k} (t) \in \mathbb{C}^s$ denotes the vector of channel gains from device $m$ to the $k$-th PS antenna, $m \in [M]$, and $\boldsymbol{z}^n_{k} (t) \in \mathbb{C}^s$ represents the additive noise at the $k$-th antenna of the PS, $n \in [N]$. 
The $i$-th entry of channel vector $\boldsymbol{h}^n_{m,k} (t)$, denoted by $h^n_{m,k,i} (t)$, is distributed according to $\mathcal{C N} \left( 0, \sigma_h^2 \right)$, $i \in [s]$, and different entries of $\boldsymbol{h}^n_{m,k} (t)$ can be correlated, while the channel gains are assumed to be independent and identically distributed (iid) across PS antennas, OFDM symbols, and wireless devices, $k \in [K]$, $n \in [N]$, $m \in [M]$.
Similarly, different entries of noise vector $\boldsymbol{z}^n_k (t)$ can be correlated, and its $i$-th entry, denoted by $z^n_{k,i} (t)$, distributed according to $\mathcal{C N} \left( 0, \sigma_z^2 \right)$, $i \in [s]$, $k \in [K]$, $n \in [N]$.
Noise vectors are also assumed to be iid across PS antennas and OFDM symbols. 
We consider the following average power constraint at each device assuming a total of $T$ global iterations:       
\begin{align}\label{AvePowerConsGen}
\frac{1}{NT} \sum\limits_{t=1}^{T} \sum\limits_{n=1}^{N} \mathbb{E} \left[ ||\boldsymbol{x}^n_{m} (t)||^2_2 \right] \le \bar{P}, \quad \forall m \in [M],
\end{align}
where the expectation is taken with respect to the randomness of the communication channel.

We assume that the devices do not have CSI, and the PS has imperfect/noisy CSI about the wireless fading MAC. 
To be precise, we assume that the PS has only imperfect CSI about the sum of the channel gains from the devices to each PS antenna, i.e., $\sum\nolimits_{m = 1}^{M} \boldsymbol{h}^n_{m,k} (t)$, $\forall k \in [K]$, for $n \in [N]$.
We denote the CSI of $\sum\nolimits_{m=1}^{M}
\boldsymbol{h}^n_{m,k} (t)$ at the PS by $\wh{\boldsymbol{h}}^n_{k} (t)$, where \cite{WeberImperctCSIMIMOTCOM}
\begin{align}\label{ImperfectCSIPSModel}
\wh{\boldsymbol{h}}^n_{k} (t) = \sum\limits_{m=1}^{M} \boldsymbol{h}^n_{m,k} (t) + \tilde{\boldsymbol{h}}^n_{k} (t), \quad \forall n,k,t,    
\end{align}
where $\tilde{\boldsymbol{h}}^n_{k} (t)$ represents the independent CSI estimation error vector with each entry an iid random variable with zero-mean and variance $\tilde{\sigma}_h^2$.
At each global iteration, the goal at the PS is to estimate $\Delta \boldsymbol{\theta} (t)$, denoted by $\Delta \wh{\boldsymbol{\theta}} \left( t \right)$, based on its received symbols $\boldsymbol{y}^n_k (t)$, and the CSI $\wh{\boldsymbol{h}}^n_{k} (t)$, $\forall n,k$. 
The PS then updates the global model as
\begin{align}\label{ParallelSGDModelUpdateNoisy}
\boldsymbol{\theta} (t+1) =\boldsymbol{\theta} (t) + \Delta \wh{\boldsymbol{\theta}} \left( t \right),   
\end{align}
and shares the new global model with the devices accurately.

\begin{remark}
The CSI of $\sum\nolimits_{m = 1}^{M} \boldsymbol{h}^n_{m,k} (t)$, $\forall n,k$, dictates that the PS only needs to estimate the sum of the channel gains from all the devices to each antenna rather than each individual channel gain $\boldsymbol{h}^n_{m,k} (t)$. 
This significantly reduces the overhead of channel estimation, particularly for a relatively large number of devices $M$ or number of PS antennas $K$, and this overhead does not increase with $M$. 
\end{remark}

We note that the PS is interested in the average of the local model updates computed by the devices rather than each individual model update. 
Motivated by the additive nature of the wireless MAC, we consider an analog approach similarly to \cite{MohammadDenizDSGDCS,KaibinParallelWork,YangFedLearOverAirComp,FLTWCMohammadDenizFading}, where the devices transmit their gradient estimates in an analog fashion without employing any channel coding.

\section{Analog FEEL without CSIT}\label{SecPeoposedAnalog}

%Analog transmission from the workers is motivated by the fact that the wireless MAC superposes the signals sent from different workers. Thus, it can provide the PS with the average of gradient estimates used to update the model parameter.   

Next, we present the proposed analog FEEL scheme in the absence of CSIT at the devices. 
For the global model update at the PS, we first assume perfect CSI at the PS, in which case $\tilde{\sigma}^2_h =0$, and study the impact of the imperfect CSI at the PS in the next subsection.

At the global iteration $t$, device $m$ aims to transmit its local update ${\Delta \boldsymbol{\theta}}_m \left( t \right) \in \mathbb{R}^{d}$ over $N = \left\lceil {d/2s} \right\rceil$ OFDM symbols across $s$ subchannels in an uncoded manner, $m \in [M]$. We denote the $i$-th entry of ${\Delta \boldsymbol{\theta}}_m \left( t \right)$ by ${\Delta {\theta}}_{m, i} \left( t \right)$, $i \in [d]$, and define, for $n \in [N]$, $m \in [M]$,
\begin{subequations}
\label{gnmESADSGDDef}
\begin{align}\label{gnmRealESADSGDDef}
\Delta {\boldsymbol{\theta}}^{n, {\rm{re}}}_{m} \left( t \right)& \triangleq [ \Delta {\theta}_{m,2(n-1)s+1} \left(t \right), \dots, \Delta {{\theta}}_{m,(2n-1)s} \left( t \right)]^T,\\
\Delta {\boldsymbol{\theta}}^{n, \rm{im}}_{m} \left( t \right)& \triangleq [ \Delta {\theta}_{m,(2n-1)s+1} \left( t \right), \dots, \Delta {\theta}_{m,2ns} \left( t \right)]^T,
\label{gnmImagESADSGDDef}\\
\Delta {\boldsymbol{\theta}}^n_{m} \left( t \right) & \triangleq \Delta {\boldsymbol{\theta}}^{n, {\rm{re}}}_{m} \left( t \right) + j \Delta {\boldsymbol{\theta}}^{n, {\rm{im}}}_{m} \left( t \right),
\label{gnmRealImagESADSGDDef}
\end{align}
\end{subequations}
where $j \triangleq \sqrt{-1}$, and we zero-pad $\Delta {\boldsymbol{\theta}}_m \left( t \right)$ to have length $2sN$. The $i$-th entry of $\Delta {\boldsymbol{\theta}}^n_{m} \left( t \right)$ is then given by
\begin{align}\label{ithgmn}
\Delta \theta^n_{m,i} \left( t \right) = \Delta \theta_{m,2(n-1)s+i} & \left( t  \right) + j \Delta \theta_{m,(2n-1)s+i} \left( t \right), \; \mbox{for $i \in [s]$, $n \in [N]$, $m \in [M]$}.
\end{align}
According to \eqref{gnmESADSGDDef}, we have
\begin{align}\label{gmwrtgmnReIm}
\Delta {\boldsymbol{\theta}}_m \left( t \right) = & \big[ \Delta {\boldsymbol{\theta}}^{1, {\rm{re}}}_{m} \left( t \right), \Delta {\boldsymbol{\theta}}^{1, {\rm{im}}}_{m} \left( t \right), \dots,  \Delta {\boldsymbol{\theta}}^{N, {\rm{re}}}_{m} \left( t \right), \Delta {\boldsymbol{\theta}}^{N, {\rm{im}}}_{m} \left( t \right) \big]^T,
\end{align}
with $N = \left\lceil {d/2s} \right\rceil$. At the $n$-th OFDM symbol of iteration $t$, device $m$ sends 
\begin{align}\label{workermSends}
\boldsymbol{x}^n_{m} (t) = \alpha_t \Delta \boldsymbol{\theta}^n_{m} (t), \quad n \in [N], m \in [M],
\end{align}
where $\alpha_t$ is the scaling factor that will be chosen according to the power constraint. 
Accordingly, the average transmit power depends on $\alpha_t$, and is evaluated as follows: 
\begin{align}\label{AvePowerConsWorkerm}
\frac{1}{NT} \sum\limits_{t=1}^{T} \alpha_t^2 \sum\limits_{n=1}^{N} ||\Delta \boldsymbol{\theta}^n_{m} (t)||^2_2 \le \bar{P}. 
\end{align}

The PS observes the following signal at its $k$-th antenna, for $k \in [K], n \in [N]$:
\begin{align}\label{ReceivedVectorPSScheAntennak}
\boldsymbol{y}^n_k (t) = \alpha_t \sum\limits_{m = 1}^{M} \boldsymbol{h}^n_{m,k} (t) \circ \Delta \boldsymbol{\theta}^n_{m} (t) + \boldsymbol{z}^n_k (t). 
\end{align}

\subsection{Perfect CSI at the PS}
In this subsection, we assume that the PS has access to perfect CSI about the sum of the channel gains from all the devices to each antenna, i.e., $\tilde{\sigma}^2_h =0$ and $\wh{\boldsymbol{h}}^n_{k} (t) = \sum\nolimits_{m=1}^{M} \boldsymbol{h}^n_{m,k} (t)$, $\forall n,k,t$.
Having access to perfect CSI, the PS combines the signals at different antennas in the following form:
\begin{align}\label{ReceivedVectorPSScheCombAntennas}
\boldsymbol{y}^n (t) \triangleq \frac{1}{K} \sum\limits_{k=1}^{K} \Big( \sum\limits_{m = 1}^{M} \boldsymbol{h}^n_{m,k} (t) \Big)^{*} \circ \boldsymbol{y}^n_k (t), 
\end{align}
whose $i$-th entry is given by
\begin{align}\label{ReceivedVectorPSScheCombAntennasith}
y^n_{i} (t) = \frac{1}{K} \sum\limits_{k=1}^{K}  \sum\limits_{m = 1}^{M} \left( {h}^n_{m,k,i} (t) \right)^{*} {y}^n_{k,i} (t), 
\end{align}
where ${y}^n_{k,i} (t)$ denotes the $i$-th entry of $\boldsymbol{y}^n_{k} (t)$, $i \in [s]$, $n \in [N]$. By substituting ${y}^n_{k,i} (t)$, given in \eqref{ReceivedVectorPSScheAntennak}, it follows that
\begin{align}\label{ReceivedVectorPSScheCombAntennasReWrith}
{y}^n_i (t) = & \underbrace{\alpha_t \sum\limits_{m=1}^{M} \Big( \frac{1}{K} \sum\limits_{k=1}^{K} \left| {h}^n_{m,k,i} (t) \right|^2 \Big)  \Delta {\theta}^n_{m,i} (t)}_{\text{\normalfont signal term}} \nonumber\\
&  + \underbrace{\frac{\alpha_t}{K}  \sum\limits_{m=1}^{M} \sum\limits_{m'=1, m' \ne m}^{M} \sum\limits_{k=1}^{K} \left( {h}^n_{m,k,i} (t) \right)^{*} {h}^n_{m',k, i} (t) \Delta {\theta}^n_{m', i} (t)}_{\text{\normalfont interference term}} \nonumber\\
&  +  \underbrace{\frac{1}{K} \sum\limits_{m=1}^{M} \sum\limits_{k=1}^{K} \left( {h}^n_{m,k, i} (t) \right)^{*} z_{k,i}^n (t)}_{\text{\normalfont noise term}}. 
\end{align}
As we can see in \eqref{ReceivedVectorPSScheCombAntennasReWrith}, ${y}^n_i (t)$ consists of three terms, specified as the signal, interference, and noise components, respectively.
Following the law of large numbers, as the number of antennas at the PS $K \to \infty$, the signal term approaches 
\begin{align}\label{SigTermApproaches}
y_{i, {\rm{sig}}}^n (t) \triangleq \alpha_t \sigma_h^2 \sum\limits_{m=1}^{M} \Delta {\theta}^n_{m,i} (t), \quad i \in [s], n \in [N],
\end{align}
from which the PS can recover 
\begin{subequations}\label{PSSigTermRecReIm}
\begin{align}\label{PSSigTermRecRe}
\frac{1}{M} \sum\limits_{m=1}^{M} \Delta {\theta}_{m,2(n-1)s+i} \left( t \right) &= \frac{ {\rm{Re}} \left\{ y_{i, {\rm{sig}}}^n (t) \right\} }{\alpha_t M \sigma_h^2},\\
\frac{1}{M} \sum\limits_{m=1}^{M} \Delta {\theta}_{m,(2n-1)s+i} \left( t  \right) &= \frac{ {\rm{Im}} \left\{ y_{i, {\rm{sig}}}^n (t) \right\} }{\alpha_t M \sigma_h^2}.\label{PSSigTermRecIm}
\end{align}
\end{subequations}
However, the interference term in \eqref{ReceivedVectorPSScheCombAntennasReWrith} does not allow the exact recovery of $\frac{1}{M} \sum\nolimits_{m=1}^{M} \Delta {\theta}_{m,2(n-1)s+i} \left( t \right)$ and $\frac{1}{M} \sum\nolimits_{m=1}^{M} \Delta {\theta}_{m,(2n-1)s+i} \left( t \right)$ from ${y}^n_i (t)$, which is observed at the PS. 
To analyze the interference term, defined as $y_{i, {\rm{itf}}}^n(t)$, we rewrite it as follows:
\begin{align}\label{IntTerAnDef}
y_{i, {\rm{itf}}}^n(t) = \alpha_t  \sum\limits_{m=1}^{M} \Big( \frac{1}{K} \sum\limits_{k=1}^{K} {h}^n_{m,k, i} (t) \sum\limits_{m'=1, m' \ne m}^{M} \left( {h}^n_{m',k,i} (t) \right)^{*} \Big) \Delta {\theta}^n_{m, i} (t), \quad i \in [s], n \in [N]. 
\end{align}
We then define, for $m \in [M], i \in [s], n \in [N]$,
\begin{align}
\mathfrak{h}_{m,i}^n (t) \triangleq \frac{1}{K} \sum\limits_{k=1}^{K} {h}^n_{m,k, i} (t) \sum\limits_{m'=1, m' \ne m}^{M} \left( {h}^n_{m',k,i} (t) \right)^{*},   
\end{align}
and is easy to verify that the mean and the variance of $\mathfrak{h}_{m,i}^n (t)$ are given by
\begin{subequations}\label{MeanVarMathFrakh}
\begin{align}\label{MeanMathFrakh}
\mathbb{E} \left[ \mathfrak{h}_{m,i}^n (t) \right] =& 0,\\ 
\mathbb{E} \left[ \left| \mathfrak{h}_{m,i}^n (t) \right|^2 \right] =& \frac{(M-1) \sigma_h^4}{K},\label{VarMathFrakh}  
\end{align}
\end{subequations}
respectively. 
We note that the local updates computed at each iteration are independent of the channel realizations experienced during the same iteration. 
From the analysis in \eqref{MeanVarMathFrakh}, we conclude that the interference term in \eqref{ReceivedVectorPSScheCombAntennasReWrith} has zero-mean and $M$ terms, each with a variance that scales with $(M-1) / K$. 
Thus, for a fixed number of wireless devices $M$, the variance of the interference term in \eqref{ReceivedVectorPSScheCombAntennasReWrith} approaches zero as $K \to \infty$. 
In practice, it is feasible to employ a sufficiently large number of antennas at the PS exploiting massive multiple-input multiple-output (MIMO) systems \cite{RusekScaleUpMassiveMIMO}. 
Numerical results with a finite number of antennas will be presented in Section \ref{SecExperiments}.

According to the above analysis, the PS estimates $\frac{1}{M} \sum\nolimits_{m=1}^{M} \Delta {\theta}_{m,2(n-1)s+i} \left( t \right)$ and $\frac{1}{M} \sum\nolimits_{m=1}^{M} \Delta {\theta}_{m,(2n-1)s+i} \left( t \right)$, for $i \in [s]$, $n \in [N]$, through
\begin{subequations}\label{PSSigTermRecReImEst}
\begin{align}\label{PSSigTermRecReEst}
\Delta \hat{\theta}_{2(n-1)s+i} \left( t \right) &= \frac{ {\rm{Re}} \left\{ y_{i}^n (t) \right\} }{\alpha_t M \sigma_h^2},\\
\Delta \hat{\theta}_{(2n-1)s+i} \left( t \right) &= \frac{ {\rm{Im}} \left\{ y_{i}^n (t) \right\} }{\alpha_t M \sigma_h^2},\label{PSSigTermRecImEst}
\end{align}
\end{subequations}
respectively. 
It then utilizes the estimated vector $\Delta \wh{\boldsymbol{\theta}} ( t) \triangleq \big[ \Delta \hat{\theta}_{1} \left( t \right), \dots, \Delta \hat{\theta}_{d} \left( t \right) \big]^T$, which is an unbiased estimate of the average of the local model updates, to update the global model as
\begin{align}\label{GlobalModelUpdatePerfectCSISubSec}
\boldsymbol{\theta} (t+1) =  \boldsymbol{\theta} (t) + \Delta \wh{\boldsymbol{\theta}} ( t).  
\end{align}
%Algorithm \ref{MUltiAntennaA_DSGD} presents the proposed analog DSGD scheme without the CSI knowledge at the wireless devices. In Algorithm \ref{MUltiAntennaA_DSGD}, for ease of presentation, we represent the estimated signals at the PS in vector form, in which $\hat{\boldsymbol{g}}_{\rm{re}}^n \left(\boldsymbol{\theta}_{t}  \right)$ and $\hat{\boldsymbol{g}}_{\rm{im}}^n \left(\boldsymbol{\theta}_{t}  \right)$ represent the estimates of $\frac{1}{M} \sum\nolimits_{m=1}^{M} {\boldsymbol{g}}^n_{m, {\rm{re}}} \left(\boldsymbol{\theta}_{t} \right)$ and $\frac{1}{M} \sum\nolimits_{m=1}^{M} {\boldsymbol{g}}^n_{m, {\rm{im}}} \left(\boldsymbol{\theta}_{t} \right)$, respectively, at the PS, $n \in [N]$.   

\subsection{Imperfect CSI at the PS} 
We now generalize the above beamforming technique by considering imperfect CSI at the PS. 
Let $\wh{\boldsymbol{h}}^n_{k} (t) = [\hat{h}^n_{k,1} (t), \dots, \hat{h}^n_{k,s} (t)]^T$ and $\tilde{\boldsymbol{h}}^n_{k} (t) = [\tilde{h}^n_{k,1} (t), \dots, \tilde{h}^n_{k,s} (t)]^T$. 
%\begin{subequations}
%\begin{align}
%\wh{\boldsymbol{h}}^n_{k} (t) &= [\hat{h}^n_{k,1} (t), \dots, \hat{h}^n_{k,s} (t)]^T,\\
%\tilde{\boldsymbol{h}}^n_{k} (t) &= [\tilde{h}^n_{k,1} (t), \dots, \tilde{h}^n_{k,s} (t)]^T.
%\end{align}
%\end{subequations}
In the case of imperfect CSI at the PS, the received signals at different PS antennas are combined as follows:
\begin{align}\label{ReceivedVectorPSScheCombAntennas_ImpCSI}
\boldsymbol{y}^n (t) & = \frac{1}{K} \sum\limits_{k=1}^{K} \Big( \wh{\boldsymbol{h}}^n_{k} (t) \Big)^{*} \circ \boldsymbol{y}^n_k (t) \nonumber\\
& = \frac{1}{K} \sum\limits_{k=1}^{K} \Big( \sum\limits_{m = 1}^{M} {\boldsymbol{h}}^n_{m,k} (t) \Big)^{*} \circ \boldsymbol{y}^n_k (t) + \frac{1}{K} \sum\limits_{k=1}^{K} \Big( \tilde{\boldsymbol{h}}^n_{k} (t) \Big)^{*} \circ \boldsymbol{y}^n_k (t), 
\end{align}
which generalizes the expression for the perfect CSI case given in \eqref{ReceivedVectorPSScheCombAntennas}. 
Accordingly, we have
\begin{align}\label{ReceivedSignalImperfectCSI}
{y}^n_i (t) = & \alpha_t \sum\limits_{m=1}^{M} \Big( \frac{1}{K} \sum\limits_{k=1}^{K} \left| {h}^n_{m,k,i} (t) \right|^2 \Big)  \Delta {\theta}^n_{m,i} (t) \nonumber\\
& + \frac{\alpha_t}{K}  \sum\limits_{m=1}^{M} \sum\limits_{m'=1, m' \ne m}^{M} \sum\limits_{k=1}^{K} \left( {h}^n_{m,k,i} (t) \right)^{*} {h}^n_{m',k, i} (t) \Delta {\theta}^n_{m', i} (t) \nonumber\\
& + \frac{1}{K} \sum\limits_{m=1}^{M} \sum\limits_{k=1}^{K} \left( {h}^n_{m,k, i} (t) \right)^{*} z_{k,i}^n (t) \nonumber\\
&+ \frac{\alpha_t}{K}  \sum\limits_{m=1}^{M} \sum\limits_{k=1}^{K} \left( \tilde{h}^n_{k,i} (t) \right)^{*} {h}^n_{m,k, i} (t) \Delta {\theta}^n_{m, i} (t) \nonumber\\
& + \frac{1}{K} \sum\limits_{k=1}^{K} \left( \tilde{h}^n_{k, i} (t) \right)^{*} z_{k,i}^n (t),
\end{align}
where the last two terms on the right hand side (RHS) are due to having imperfect CSI at the PS, and for $\tilde{\sigma}_h^2=0$ the above expression is equivalent to \eqref{ReceivedVectorPSScheCombAntennasReWrith}. 
We denote the extra interference term introduced because of the lack of perfect CSI at the PS by $\tilde{y}_{i, {\rm{itf}}}^n(t)$ given by
\begin{align}
\tilde{y}_{i, {\rm{itf}}}^n(t) = \alpha_t \sum\limits_{m=1}^{M} \Big( \frac{1}{K} \sum\limits_{k=1}^{K} \left( \tilde{h}^n_{k,i} (t) \right)^{*} {h}^n_{m,k, i} (t) \Big) \Delta {\theta}^n_{m, i} (t), \quad i \in [s], n \in [N].  
\end{align}
We define, for $m \in [M], i \in [s], n \in [N]$,
\begin{align}
\tilde{\mathfrak{h}}_{m,i}^n (t) \triangleq \frac{1}{K} \sum\limits_{k=1}^{K} \left( \tilde{h}^n_{k,i} (t) \right)^{*} {h}^n_{m,k, i} (t),   
\end{align}
where we have
\begin{subequations}
\begin{align}
\mathbb{E} \left[ \tilde{\mathfrak{h}}_{m,i}^n (t) \right] =& 0,\\ 
\mathbb{E} \left[ \big| \tilde{\mathfrak{h}}_{m,i}^n (t) \big|^2 \right] =& \frac{\tilde{\sigma}_h^2 \sigma_h^2}{K}. 
\end{align}
\end{subequations}
Therefore, lack of perfect CSI at the PS introduces an extra interference term with zero-mean which includes $M$ terms, each with a variance scaled with $1/K$. 
Similarly to the perfect CSI scenario, the PS estimates $\frac{1}{M} \sum\nolimits_{m=1}^{M} \Delta {\theta}_{m,2(n-1)s+i} \left( t \right)$ and $\frac{1}{M} \sum\nolimits_{m=1}^{M} \Delta {\theta}_{m,(2n-1)s+i} \left( t \right)$, for $i \in [s]$, $n \in [N]$, through
\begin{subequations}\label{DeltaThetaHatImperfectCSI}
\begin{align}
\Delta \hat{\theta}_{2(n-1)s+i} \left( t \right) &= \frac{ {\rm{Re}} \left\{ y_{i}^n (t) \right\} }{\alpha_t M \sigma_h^2},\\
\Delta \hat{\theta}_{(2n-1)s+i} \left( t \right) &= \frac{ {\rm{Im}} \left\{ y_{i}^n (t) \right\} }{\alpha_t M \sigma_h^2},
\end{align}
\end{subequations}
and uses $\Delta \wh{\boldsymbol{\theta}} ( t)$ to update the global model as in \eqref{GlobalModelUpdatePerfectCSISubSec}. 
%\begin{align}
%\boldsymbol{\theta} (t+1) =  \boldsymbol{\theta} (t) + \Delta \wh{\boldsymbol{\theta}} ( t).  
%\end{align}

\begin{remark}\label{RemPowerDecay}
We note that with SGD the empirical variances of the local model updates decay over time and approach zero asymptotically \cite{BottouLargeScaleSGD,ScalableDNNStorm,DCLimitedPrecisionGupta,MohammadDenizDSGDCS,UseLocalSGDLin}. 
Thus, for robust communication of the local model updates against noise at each global iteration, it is reasonable to increase the power allocation factor $\alpha_t$ over time.        
\end{remark}

\begin{remark}\label{RemCompression}
We remark that the main focus in this paper is to develop techniques for FEEL with no CSIT, as well as imperfect CSI at the PS. 
Our approach to tackle this problem is to employ multiple antennas at the PS, which can help to mitigate the effect of fading, and, in the limit, align the received signals at the PS. 
We can further employ some of the existing schemes in the literature providing more efficient communication over the limited bandwidth wireless MAC, such as the idea of linear projection proposed in \cite{MohammadDenizDSGDCS}. 
We leave the analysis of such combined techniques as future work.       
\end{remark}

\section{Convergence Analysis}\label{SecConverge}
In this section, we provide a convergence analysis of the proposed analog FEEL scheme with no CSIT and imperfect CSI at the PS. 
For ease of presentation, we consider $N=1$, i.e., $s=d/2$, and drop the dependency of all the variables on $n$. 
Accordingly, the received signal at the PS, given in \eqref{ReceivedSignalImperfectCSI}, can be rewritten as follows:
\begin{subequations}\label{Convergence_y_i_t_rewrite}
\begin{align}
{y}_i (t) = \sum\limits_{l=1}^{5} {y}_{i,l} (t), \quad \mbox{for $i \in [d/2]$},
\end{align}
where 
\begin{align}
y_{i,1} (t) \triangleq & \alpha_t \sum\limits_{m=1}^{M} \Big( \frac{1}{K} \sum\limits_{k=1}^{K} \left| {h}_{m,k,i} (t) \right|^2 \Big)  \left( \Delta {\theta}_{m,i} (t) + j \Delta {\theta}_{m,d/2+i} (t) \right), \\
y_{i,2} (t) \triangleq & \frac{\alpha_t}{K}  \sum\limits_{m=1}^{M} \sum\limits_{m'=1, m' \ne m}^{M} \sum\limits_{k=1}^{K} \left( {h}_{m,k,i} (t) \right)^{*} {h}_{m',k, i} (t) \left( \Delta {\theta}_{m',i} (t) + j \Delta {\theta}_{m',d/2+i} (t) \right), \\
y_{i,3} (t) \triangleq & \frac{1}{K} \sum\limits_{m=1}^{M} \sum\limits_{k=1}^{K} \left( {h}_{m,k, i} (t) \right)^{*} z_{k,i} (t),\\
y_{i,4} (t) \triangleq & \frac{\alpha_t}{K}  \sum\limits_{m=1}^{M} \sum\limits_{k=1}^{K} \left( \tilde{h}_{k,i} (t) \right)^{*} {h}_{m,k, i} (t) \left( \Delta {\theta}_{m,i} (t) + j \Delta {\theta}_{m,d/2+i} (t) \right),\\
y_{i,5} (t) \triangleq & \frac{1}{K} \sum\limits_{k=1}^{K} \left( \tilde{h}_{k, i} (t) \right)^{*} z_{k,i} (t).
\end{align}
\end{subequations}
We further define, for $l \in [5]$, 
\begin{align}\label{DeltaHatTheta_l}
\Delta \hat{\theta}_{i, l} \left( t \right) \triangleq \begin{cases} 
\frac{ {\rm{Re}} \left\{ y_{i, l} (t) \right\} }{\alpha_t M \sigma_h^2}, & \mbox{if $1 \le i \le d/2$},\\
\frac{ {\rm{Im}} \left\{ y_{i-d/2, l} (t) \right\} }{\alpha_t M \sigma_h^2}, & \mbox{otherwise},
\end{cases} 
\end{align}
according to which the estimate of the average local updates at the PS can be rewritten as
\begin{align}
\Delta \hat{\theta}_{i} \left( t \right) = \sum\limits_{l=1}^{5} \Delta \hat{\theta}_{i, l} \left( t \right), \quad \mbox{for $i \in [d]$}.    
\end{align}

\subsection{Preliminaries}
Let the optimal solution minimizing the loss function $F(\boldsymbol{\theta})$ be defined as
\begin{align}
\boldsymbol{\theta}^* \triangleq \arg \mathop {\min }\limits_{\boldsymbol{\theta}} F(\boldsymbol{\theta}),  
\end{align}
and we denote the minimum value of the loss function by $F^* = F(\boldsymbol{\theta}^*)$. 
We also denote the minimum value of the local loss function $F_m$ by $F^*_m$, for $m \in [M]$. 
We further define 
\begin{align}
\Gamma \triangleq F^* -  \sum\limits_{m=1}^{M} \frac{B_m}{B} F^*_m,    
\end{align}
where we note that $\Gamma \ge 0$ captures the amount of bias in the data distribution across the devices. 
$\Gamma$ increases with the heterogeneity of data across the devices.

We use the same learning rate across different devices and local iterations during each global iteration, but allow it to change over different global iterations; 
that is, we assume $\eta_m^i(t) = \eta(t)$, $\forall m, i$. Accordingly, we have
\begin{align}\label{ConvSGDDevicem}
\boldsymbol{\theta}_m^{i+1} (t) = \boldsymbol{\theta}_m^i (t) - \eta (t) \nabla F_m \left( \boldsymbol{\theta}_m^i (t), \xi_m^i (t) \right),  \quad \mbox{$i \in [\tau]$}, \mbox{$m \in [M]$}, 
\end{align}
and 
\begin{align}\label{DeltaTheta_m_t_conver}
\boldsymbol{\theta}_m^{i+1} (t) - \boldsymbol{\theta}_m^{1} (t) = - \eta (t) \sum\limits_{l=1}^{i} \nabla F_m \left( \boldsymbol{\theta}_m^l (t), \xi_m^l (t) \right).   
\end{align}

\begin{assumption}\label{AssumpSmoothLoss}
The loss functions $F_1, \dots, F_M$ are all $L$-smooth; that is, $\forall \boldsymbol{v}, \boldsymbol{w} \in \mathbb{R}^d$, 
\begin{align}\label{ConvLSmoothCondit}
F_m(\boldsymbol{v}) - F_m(\boldsymbol{w}) \le \langle \boldsymbol{v} - \boldsymbol{w} , \nabla F_m (\boldsymbol{w}) \rangle + \frac{L}{2} \left\| \boldsymbol{v} - \boldsymbol{w} \right\|^2_2, \quad \forall m \in [M].
\end{align}
\end{assumption}

\begin{assumption}\label{AssumpStrongConvexLoss}
The loss functions $F_1, \dots, F_M$ are all $\mu$-strongly convex; that is, $\forall \boldsymbol{v}, \boldsymbol{w} \in \mathbb{R}^d$, 
\begin{align}\label{ConvMuStConvexCondit}
F_m(\boldsymbol{v}) - F_m(\boldsymbol{w}) \ge \langle \boldsymbol{v} - \boldsymbol{w} , \nabla F_m (\boldsymbol{w}) \rangle + \frac{\mu}{2} \left\| \boldsymbol{v} - \boldsymbol{w} \right\|^2_2, \quad \forall m \in [M].      
\end{align}
\end{assumption}

\begin{assumption}\label{AssumpBoundedVarGradient}
The expected squared $l_2$-norm of the stochastic gradients are bounded; that is,
\begin{align}\label{ConvNorm2Bound}
\mathbb{E}_{\xi} \left [ \left\| \nabla F_m \left( \boldsymbol{\theta}_m^i (t), \xi_m^i (t) \right) \right\|^2_2 \right] \le G^2, \quad \forall i \in [\tau], \forall m \in [M], \; \forall t.      
\end{align}
\end{assumption}

\subsection{Convergence Rate}

Here we provide the convergence rate of the proposed analog FEEL scheme with blind transmitters. The proof is provided in Appendix \ref{A_AppTheorem}.    

\begin{theorem}\label{A_Theoremtheta_thetastar}
Let $0 < \eta(t) \le \min \big\{ 1, \frac{1}{\mu \tau} \big\}$, $\forall t$. We have
\begin{subequations}\label{A_ConvTheoremtheta_thetastarKequalM}
\begin{align}\label{A_ConvTheoremtheta_thetastarKequalM_main}
\mathbb{E} \left[ \left\| \boldsymbol{\theta} (t) - {\boldsymbol{\theta}}^* \right\|_2^2 \right] \le  \Big( \prod\limits_{i=0}^{t-1} A(i) \Big) \left\| {\boldsymbol{\theta}} (0) - {\boldsymbol{\theta}}^* \right\|_2^2 + \sum\limits_{j=0}^{t-1} B(j) \prod\limits_{i=j+1}^{t-1} A(i),  
\end{align}
where 
\begin{align}\label{A_ConvTheoremtheta_thetastarKequalM_A}
A(i) \triangleq & 1 - \mu \eta (i) \left( \tau - \eta(i) (\tau - 1 ) \right),\\ 
B(i) \triangleq &  \frac{\big(1+\frac{\tilde{\sigma}_h^2}{M{\sigma}_h^2}\big) \eta^2(i) \tau^2 G^2}{K} + \frac{\big(1+\frac{\tilde{\sigma}_h^2}{M{\sigma}_h^2}\big)\sigma_z^2d}{2 \alpha_t^2 KM \sigma_h^2}  + \left( 1+ \mu (1- \eta(i)) \right) \eta^2(i) G^2 \frac{\tau (\tau-1)(2\tau-1)}{6} \nonumber\\
& + (\tau^2 + \tau-1) \eta^2(i) G^2  + 2  \eta(i) (\tau - 1) \Gamma, \label{A_ConvTheoremtheta_thetastarKequalM_B}
\end{align}
\end{subequations}
and the expectation is with respect to the stochastic gradient function and the randomness of the underlying wireless channel.
\end{theorem}
\begin{proof}
See Appendix \ref{A_AppTheorem}. 
\end{proof}

\begin{corollary}\label{CorollLossGap}
Let $0 < \eta(t) \le \min \big\{ 1, \frac{1}{\mu \tau} \big\}$, $\forall t$. Given a total number of $T$ global iterations, the $L$-smoothness of loss function $F(\cdot)$ results in 
\begin{align}\label{A_ConvF_FstarKequalM}
\mathbb{E} \left[ F( \boldsymbol{\theta} (T)) \right] - F^* \le & \frac{L}{2} \mathbb{E} \left[ \left\| \boldsymbol{\theta} (T) - {\boldsymbol{\theta}}^* \right\|_2^2 \right] \nonumber \\
\le & \frac{L}{2} \Big( \prod\limits_{i=0}^{T-1} A(i) \Big) \left\| {\boldsymbol{\theta}} (0) - {\boldsymbol{\theta}}^* \right\|_2^2 + \frac{L}{2} \sum\limits_{j=0}^{T-1} B(j) \prod\limits_{i=j+1}^{T-1} A (i), 
\end{align}
where the last inequality follows from \eqref{A_ConvTheoremtheta_thetastarKequalM_main}. 
\end{corollary}

\begin{remark}
The second term in $B(i)$, $\frac{\left(1+ \tilde{\sigma}_h^2/(M{\sigma}_h^2) \right)\sigma_z^2d}{2 \alpha_t^2 K M \sigma_h^2}$, which is the result of the additive noise over the MAC, is not scaled with $\eta(i)$.
Therefore, even for a decreasing learning rate $\eta(t)$, i.e., $\mathop {\lim }\limits_{t \to \infty } \eta(t) = 0$, we have $\mathop {\lim }\limits_{t \to \infty } B(t) = \frac{\left(1+ \tilde{\sigma}_h^2/(M{\sigma}_h^2) \right)\sigma_z^2d}{2 \alpha_t^2 K M \sigma_h^2} \ne 0$, which shows that $\mathop {\lim }\limits_{t \to \infty } \mathbb{E} \left[ F( \boldsymbol{\theta} (t)) \right] - F^* \ne 0$. 
However, we note that the destructive effect of this term in the convergence rate reduces with the number of PS antennas, $K$.
We further remark that $\frac{\tilde{\sigma}_h^2 \eta^2(i) \tau^2 G^2}{{\sigma}_h^2KM} + \frac{\tilde{\sigma}_h^2 \sigma_z^2d}{2 \alpha_t^2 K M^2 \sigma_h^4}$ captures the impact of the imperfect CSI at the PS, which also reduces with $K$.  
\end{remark}

\begin{corollary}
Consider a simplified setting $\eta(t) = \eta$, $\forall t$, and $\tau = 1$. 
Accordingly, Corollary \ref{CorollLossGap} can be simplified as
\begin{align}
\mathbb{E} \left[ F( \boldsymbol{\theta} (T)) \right] - F^* & \le \frac{L}{2} (1 - \mu \eta)^T \left\| \boldsymbol{\theta} (0) - {\boldsymbol{\theta}}^* \right\|_2^2 \nonumber\\
& + \frac{L}{2 \mu \eta} \Big( \Big( 1+\frac{\tilde{\sigma}_h^2}{M{\sigma}_h^2} \Big) \Big( \frac{\eta^2 G^2}{K} + \frac{\sigma_z^2d}{2 \alpha_t^2 M K \sigma_h^2} \Big) + \eta^2 G^2 \Big) \big( 1- (1-\mu \eta)^T \big).    
\end{align}
\end{corollary}

% \usepackage{multirow}
% \usepackage{multirow}
\begin{table}[t!]
\caption{CNN architecture for image classification on MNIST and CIFAR-10.}
\centering
\begin{tabular}{|c||c|}
\hline
MNIST & CIFAR-10 \\ \specialrule{.2em}{.01em}{.01em}
\multirow{3.75}{*}{\begin{tabular}[c]{@{}c@{}}5 $\times$ 5 convolutional layer, 32 channels, \\ ReLU activation, same padding\end{tabular}}        & \begin{tabular}[c]{@{}c@{}}3 $\times$ 3 convolutional layer, 32 channels, \\ ReLU activation, same padding\end{tabular}  \\ \cline{2-2} & \begin{tabular}[c]{@{}c@{}}3 $\times$ 3 convolutional layer, 32 channels, \\ ReLU activation, same padding\end{tabular}  \\ \cline{2-2} & 2 $\times$ 2 max pooling \\ \hline
\multirow{3}{*}{2 $\times$ 2 max pooling} & dropout with probability 0.2 \\ \cline{2-2} & \begin{tabular}[c]{@{}c@{}}3 $\times$ 3 convolutional layer, 64 channels, \\ ReLU activation, same padding\end{tabular}  \\ \hline
\multirow{3}{*}{\begin{tabular}[c]{@{}c@{}}5 $\times$ 5 convolutional layer, 64 channels, \\ ReLU activation, same padding\end{tabular}}        & \begin{tabular}[c]{@{}c@{}}3 $\times$ 3 convolutional layer, 64 channels, \\ ReLU activation, same padding\end{tabular}  \\ \cline{2-2} & 2 $\times$ 2 max pooling \\ \cline{2-2} & dropout with probability 0.3 \\ \hline
\multirow{3}{*}{2 $\times$ 2 max pooling} & \begin{tabular}[c]{@{}c@{}}3 $\times$ 3 convolutional layer, 128 channels, \\ ReLU activation, same padding\end{tabular} \\ \cline{2-2} & \begin{tabular}[c]{@{}c@{}}3 $\times$ 3 convolutional layer, 128 channels, \\ ReLU activation, same padding\end{tabular} \\ \hline
\multirow{2}{*}{\begin{tabular}[c]{@{}c@{}}fully connected layer with 1024 units,\\ ReLU activation, dropout with probability 0.2\end{tabular}} & 2 $\times$ 2 max pooling \\ \cline{2-2} & dropout with probability 0.4  \\ \hline
\multicolumn{2}{|c|}{softmax output layer with 10 units}         \\ \hline
\end{tabular}
\label{TableCNNArchit}
\end{table}

\section{Numerical Experiments}\label{SecExperiments}
Here we evaluate the performance of the proposed analog FEEL algorithm with no CSI available at the wireless devices. 
We are particularly interested in investigating the impact of the number of PS antennas on the performance. 
We perform image classification on MNIST \cite{LeCunMNIST} and CIFAR-10 datasets \cite{Cifar10DatasetKrizhevsky} using ADAM optimizer \cite{ADAMDC}. 
We train different convolutional neural networks (CNNs) whose architectures are described in Table \ref{TableCNNArchit}.   
The performance is measured as the accuracy with respect to the test dataset, known as the \textit{test accuracy}, versus the global iteration count, $t$. 
%We train the network for $T=800$ iterations. 

We consider two data distribution scenarios across the devices. 
In the non-iid data distribution scenario, we split the training data samples with the same label/class to $M/10$ disjoint groups (assuming that $M$ is divisible by $10$). 
Thus, having 10 labels/classes for both MNIST and CIFAR-10 datasets, this results in $M$ disjoint training datasets, each consisting of samples with the same label/class, and we assign each group to a distinct device.
On the other hand, in the iid data distribution scenario, we randomly split the training dataset into $M$ disjoint datasets, and assign each of them to a distinct device.
%We consider non-iid and iid data distribution while using MNIST and CIFAR-10, respectively, for training. 
We set the local mini-batch sample size to $\left| \xi_m^i (t) \right| = 500$, $\forall m, i, t$, for each experiment.

We consider $M=20$ wireless devices in the system.  
For simplicity, we assume that the $s$ channel gains associated with each OFDM symbol from each device to each PS antenna are iid, and $\sigma_h^2 = 1$.
For each experiment, we measure the test accuracy for $T=400$ global iterations, and we set the power allocation factor at the devices to $\alpha_t = 1 + 10^{-3}t$, $t \in [T]$. 
We further assume that $s = d/2$ resulting in $N=1$. 
We note that, for a fixed power allocation factor $\alpha_t$, $\forall t$, the value of $s$ does not have any impact on the accuracy of the proposed analog FEEL scheme; 
instead, any change in $s$ scales the average transmit power, whose value is proportional to $N$.
For the experiments, we assume that the CSI estimation error at the PS, i.e., $\tilde{h}^n_{k,i} (t)$, is distributed according to $\mathcal{C N} \left( 0, \tilde{\sigma}_h^2 \right)$, $\forall k, i, n, t$.
%The performance is measured as the accuracy with respect to the test samples based on the updated model parameters at each DSGD iteration.    

For numerical comparison, we also consider a benchmark, in which the PS receives the average of the local model updates ${\Delta \boldsymbol{\theta}}(t) = \frac{1}{M} \sum\nolimits_{m=1}^{M} {\Delta \boldsymbol{\theta}}_m \left( t \right)$ from the devices in an error-free manner, and updates the global model according to this noiseless observation at each iteration. 
We refer to this as the \textit{error-free shared link} scenario, and its accuracy can serve as an upper bound on the performance of the proposed analog FEEL scheme.  
%For a fair comparison, we assume that $\frac{1}{M} \sum\nolimits_{m=1}^{M} \boldsymbol{g}_m \left(\boldsymbol{\theta}_{t} \right)$ is received at the PS after $N = \left\lceil {d/2s} \right\rceil$ time slots.

\begin{figure}[t!]
\centering
\begin{subfigure}{.5\textwidth}
  \centering
  \includegraphics[scale=0.55,trim={16pt 5pt 43pt 31pt},clip]{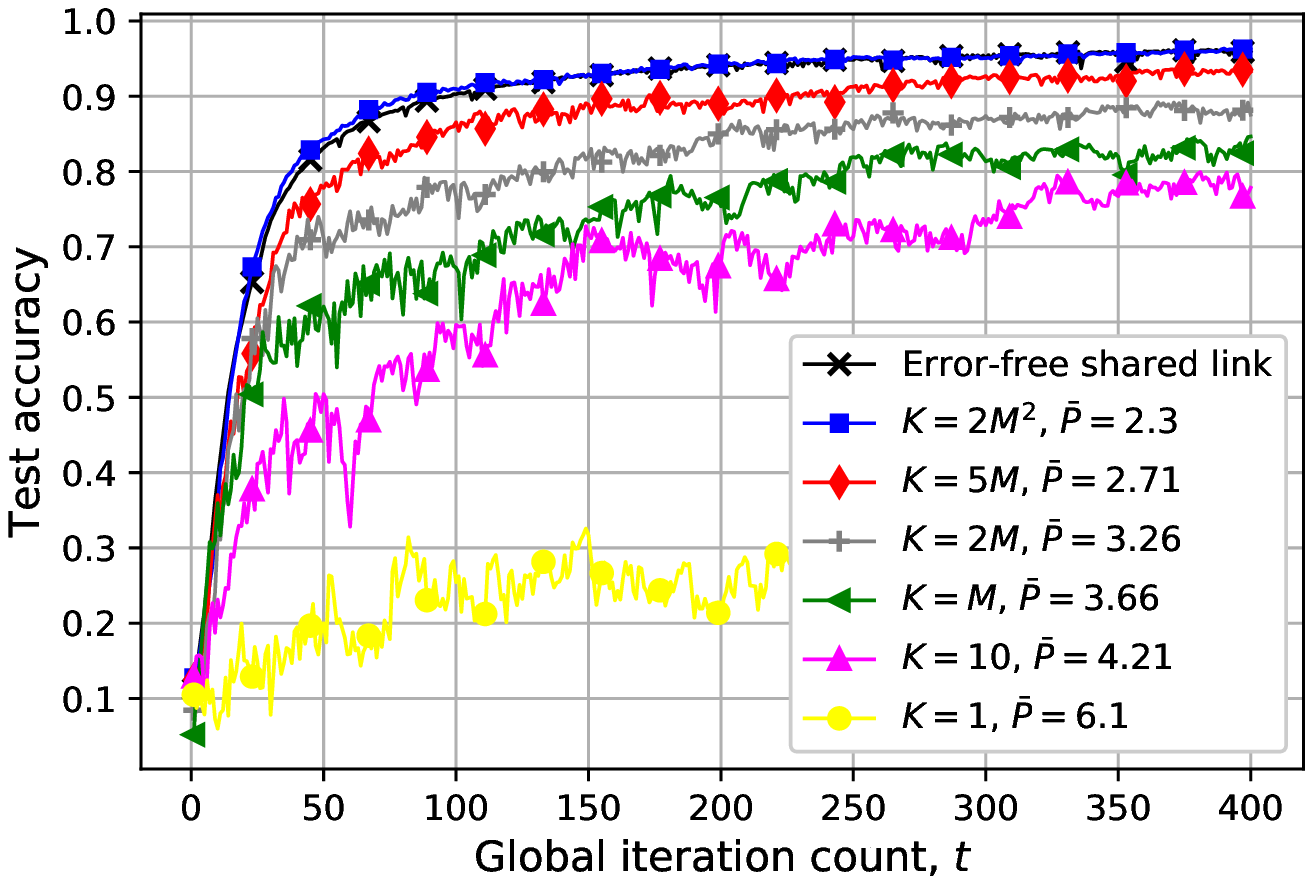}\vspace{0cm}
  \caption{Perfect CSI at PS, $\big(\sigma^2_z, \tilde{\sigma}_h^2\big) = (10, 0)$}
  \label{Fig_nonIID_CSI_Lownoise}
\end{subfigure}%
\begin{subfigure}{.5\textwidth}
  \centering
  \includegraphics[scale=0.55,trim={16pt 5pt 43pt 31pt},clip]{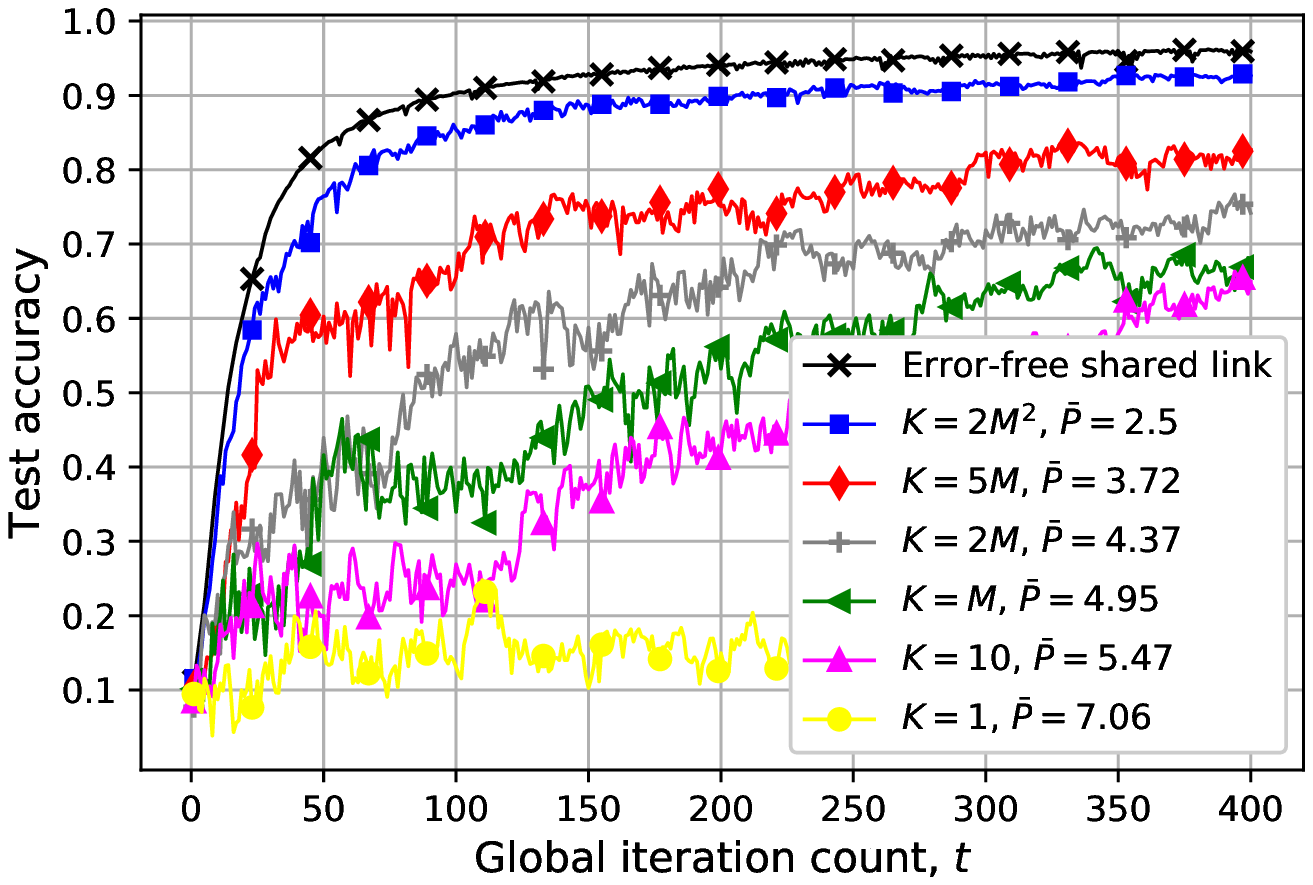}\vspace{0cm}
  \caption{Perfect CSI at PS, $\big(\sigma^2_z, \tilde{\sigma}_h^2\big) = (50, 0)$}
  \label{Fig_nonIID_CSI_Highnoise}
\end{subfigure}\\\vspace{.5cm}
\begin{subfigure}{.49\textwidth}
  \centering
  \includegraphics[scale=0.55,trim={16pt 5pt 43pt 31pt},clip]{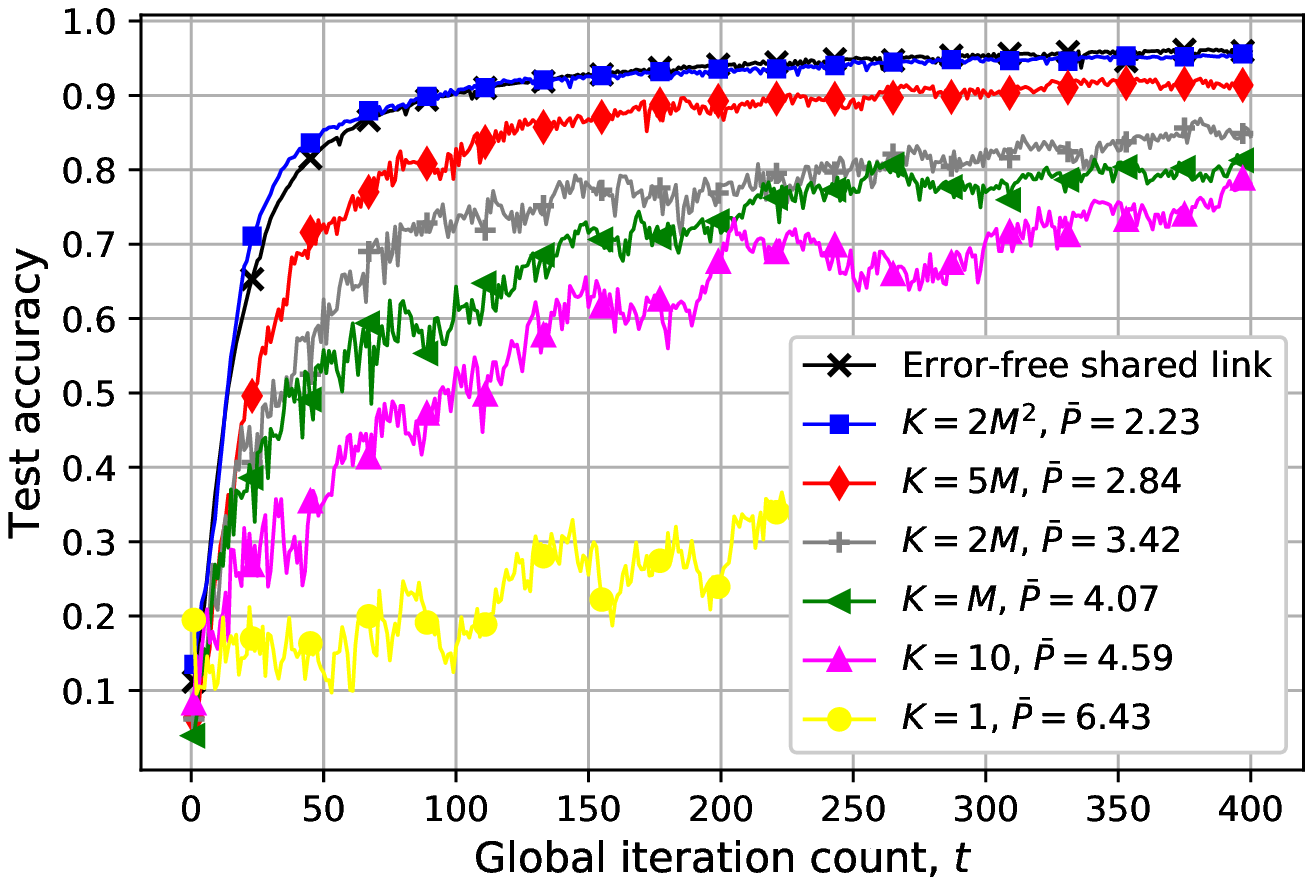}\vspace{0cm}
  \caption{Imperfect CSI at PS, $\big(\sigma^2_z, \tilde{\sigma}_h^2\big) = (10, M{\sigma}_h^2/2)$}
  \label{Fig_nonIID_imCSI_Highnoise}
\end{subfigure}
\begin{subfigure}{.49\textwidth}
  \centering
  \includegraphics[scale=0.55,trim={16pt 5pt 43pt 31pt},clip]{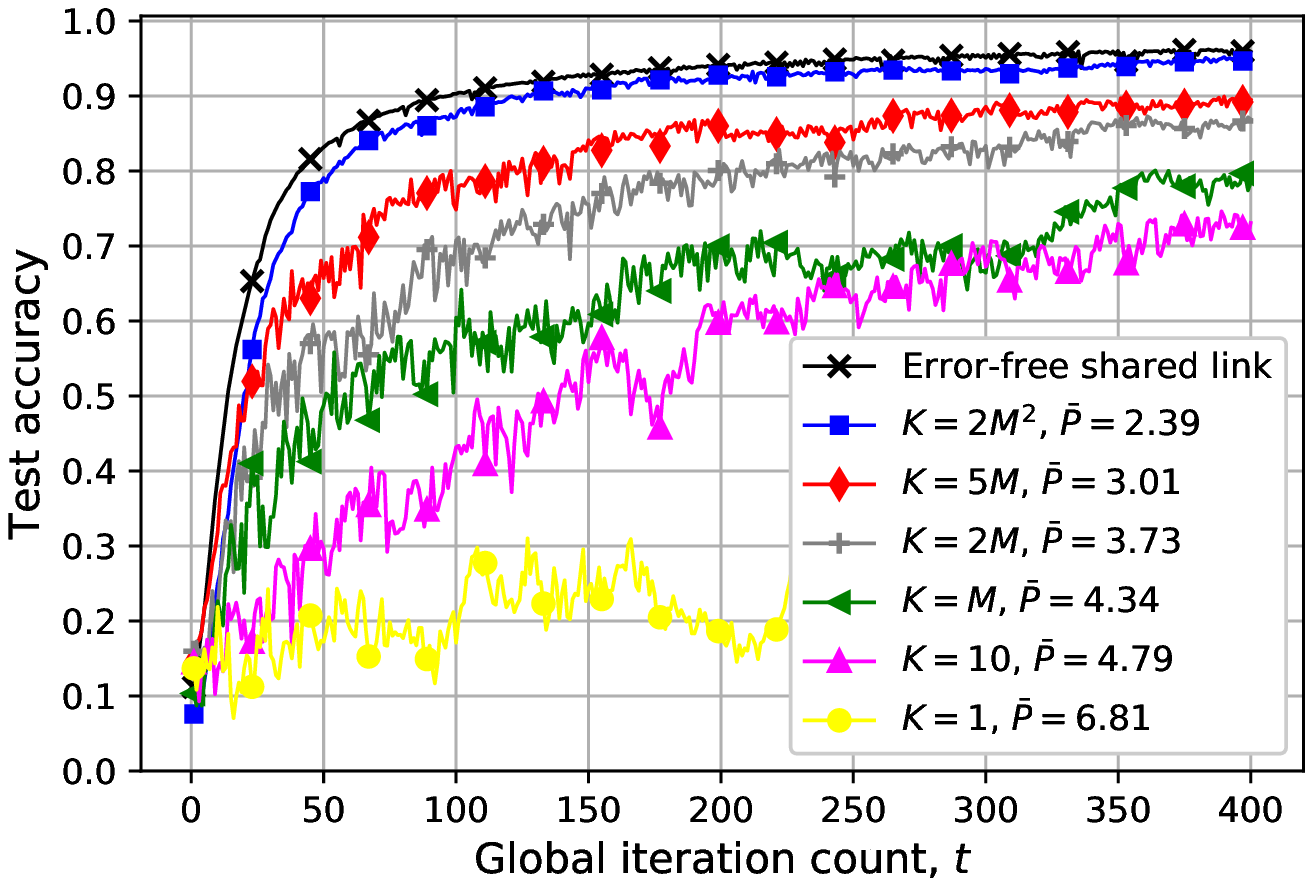}\vspace{0cm}
  \caption{Imperfect CSI at PS, $\big(\sigma^2_z, \tilde{\sigma}_h^2\big) = (10, M{\sigma}_h^2)$}
  \label{Fig_nonIID_imCSI_Highnoise_2}
\end{subfigure}
%\vspace{-.15cm}
\caption{Test accuracy of the proposed analog FEEL algorithm for non-iid MNIST data with different number of antennas $K \in \{ 1,10,M,2M,5M,2M^2 \}$ for $M=20$, $\sigma_h^2 = 1$ $\tau = 3$, and $\left| \xi_m^i (t) \right| = 500$, $\forall m, i, t$.}
\label{Fig_nonIID_CSI_imCSI}
\end{figure}

In Fig. \ref{Fig_nonIID_CSI_imCSI} we illustrate the performance of the proposed analog FEEL scheme with no CSIT for increasing number of PS antennas, $K \in \{ 1,10,M,2M,5M,2M^2 \}$, with non-iid MNIST data distributed across the devices, and number of local iterations $\tau = 3$. 
In Figs. \ref{Fig_nonIID_CSI_Lownoise} and \ref{Fig_nonIID_CSI_Highnoise} we assume perfect CSI at the PS, and investigate the performance for an increase in the noise variance from $\sigma_z^2 = 10$ to $\sigma_z^2 = 50$.
We also include the performance of the error-free shared link scenario. 
As can be seen, both the final test accuracy and the convergence speed increases with the number of PS antennas, with the improvement significantly more noticeable when the noise level is higher. 
This is due to the fact that increasing $K$ mitigates the effects of both the interference and noise terms, inferred from \eqref{ReceivedVectorPSScheCombAntennasReWrith}. 
Thus, the advantage of having more PS antennas is more pronounced when the channel is noisier. 
For example, for $\sigma_z^2 = 10$, the proposed scheme with $K = 2M^2$ antennas at the PS and average power $\bar{P}=2.3$ performs as well as the error-free shared link scenario.
On the other hand, further reducing the average signal-to-noise ratio $\bar{P}/{\sigma_{z}^2}$ by setting $\sigma_z^2 = 50$ results in a small performance gap between the error-free shared link scenario and the proposed scheme with $K = 2M^2$.
These results illustrate the success of the proposed scheme with sufficient number of PS antennas in mitigating the noise term even when the average signal-to-noise ratio $\bar{P}/{\sigma_{z}^2}$ is as small as $0.05$.
%{\color{red}According to the observations, exploiting a sufficiently large number of PS antennas may provide an order-optimal accuracy performance with respect to the error-free shared link scenario even for relatively small $\bar{P}/{\sigma_{z}^2}$ values.} 
Surprisingly, the accuracy improves drastically even with a few antennas at the PS, e.g., $K=10$. 
We note that, with all the other parameters fixed, the required average transmit power reduces with $K$, which verifies a faster convergence rate with higher $K$ resulting in a faster reduction in the empirical variances of the local model updates over time. 
The same observation is made by reducing $\sigma_z^2$ from $50$ to $10$ while all the other parameters are fixed.

\begin{figure}[t!]
\centering
\begin{subfigure}{.5\textwidth}
  \centering
  \includegraphics[scale=0.55,trim={16pt 5pt 43pt 31pt},clip]{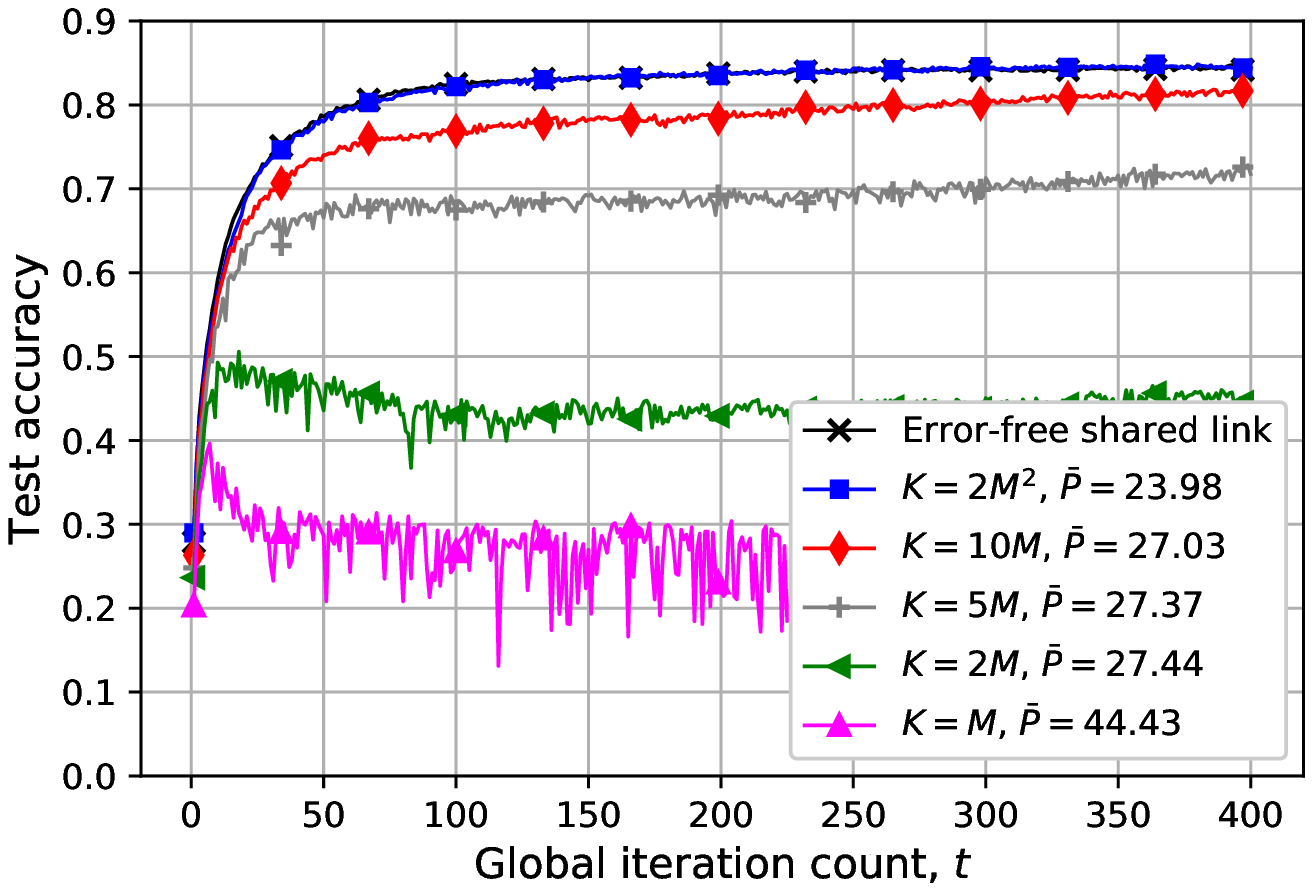}
  \caption{Perfect CSI at PS, $\tilde{\sigma}_h^2 = 0$}
  \label{Fig_IID_CSI_Lownoise}
\end{subfigure}%
\begin{subfigure}{.5\textwidth}
  \centering
  \includegraphics[scale=0.55,trim={16pt 5pt 43pt 31pt},clip]{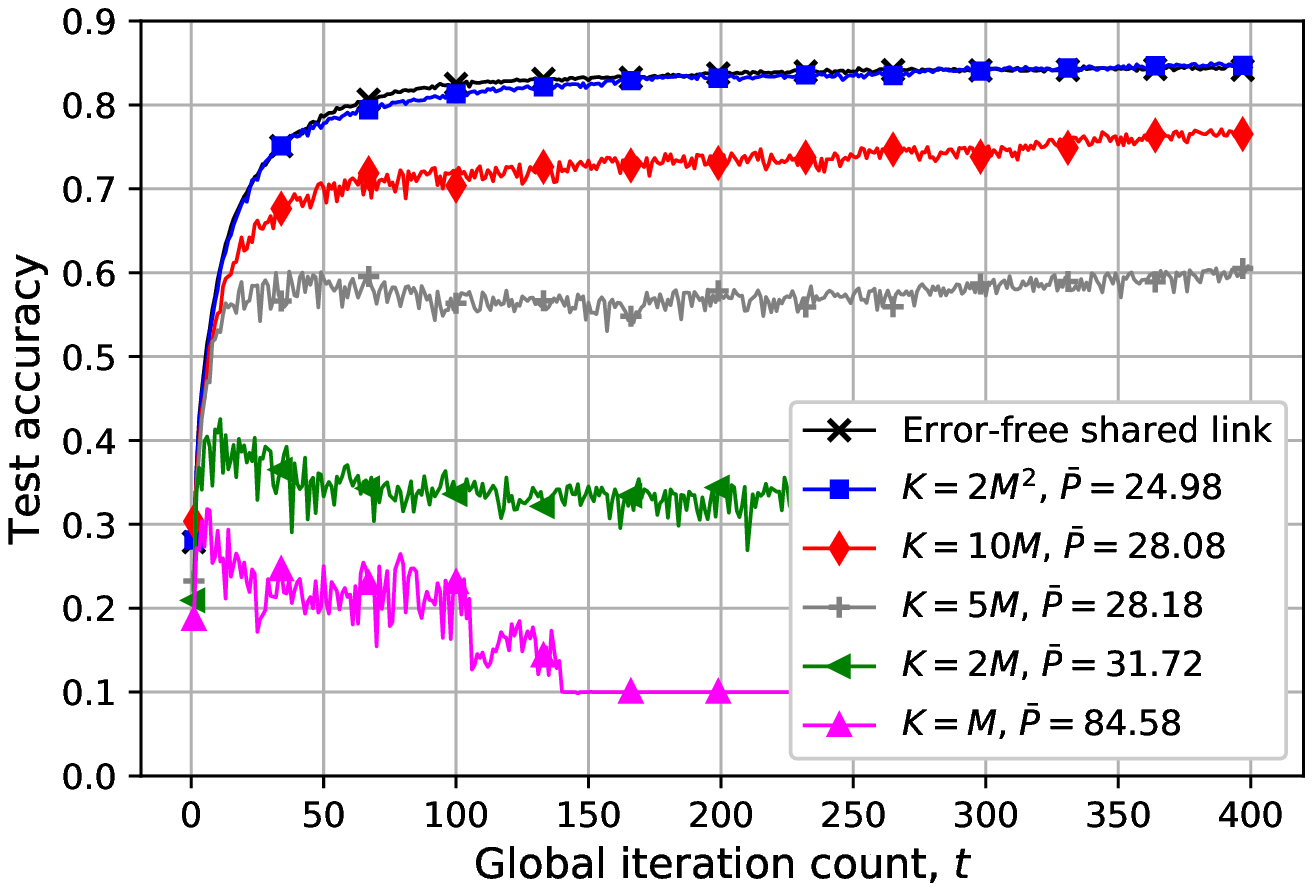}
  \caption{Imperfect CSI at PS, $\tilde{\sigma}_h^2 = M \sigma_h^2/2$}
  \label{Fig_IID_imCSI_Highnoise}
\end{subfigure}
\caption{Test accuracy of the proposed analog FEEL algorithm for iid CIFAR-10 data with different number of antennas $K \in \{ M,2M, 5M, 10M, 2M^2 \}$ for $M=20$, $\sigma_h^2 = \sigma_z^2 = 1$, $\tau = 5$, and $\left| \xi_m^i (t) \right| = 500$, $\forall m, i, t$.}
\label{Fig_IID_imCSI}
\end{figure}

Similar observations can be made in Figs. \ref{Fig_nonIID_imCSI_Highnoise} and \ref{Fig_nonIID_imCSI_Highnoise_2} considering imperfect CSI at the PS with $\tilde{\sigma}_h^2 = M{\sigma}_h^2/2$ and $\tilde{\sigma}_h^2 = M{\sigma}_h^2$, respectively. 
We observe the additional benefits of a large number of PS antennas in mitigating the adverse effects of imperfect CSI at the PS.
Comparing the two figures, we can see that the benefits are more highlighted when the variance of the CSI estimation error is larger. 
Even when the variance of the CSI error is the same as that of the sum of channel gains from the devices, i.e., when $\tilde{\sigma}_h^2 = M {\sigma}_h^2$, the proposed scheme with a sufficient number of PS antennas performs almost as well as the error-free shared link scenario.
Therefore, the proposed analog FEEL scheme can alleviate the negative effects of both the lack of CSIT and the imperfect CSI at the PS.

In Fig. \ref{Fig_IID_imCSI}, we investigate the performance of the proposed analog FEEL scheme with no CSIT for the more challenging CIFAR-10 dataset, distributed in an iid manner across the devices, considering different $K$ values, $K \in \{ M,2M, 5M, 10M, 2M^2 \}$, with $\tau = 5$ local iteration steps.  
Similarly to Fig. \ref{Fig_nonIID_CSI_imCSI}, we observe that the performance of the proposed scheme improves significantly with the number of PS antennas, and the improvement is more pronounced when the CSI is imperfect at the PS. 
In both cases under consideration, the proposed scheme with $K=2M^2$ antennas at the PS provides a performance as well as that for the benchmark error-free shared link scenario. 
However, the average required power in the experiments with CIFAR-10 dataset is higher than that for MNIST; 
this is mainly because of the larger network architecture required to reach reasonable accuracy levels for CIFAR-10, which leads to the gradients with higher norms, and consequently, resulting in higher empirical variance for the local model updates.
Furthermore, the gap between the performance of the proposed analog FEEL scheme for different $K$ values is larger than that observed in Fig. \ref{Fig_nonIID_CSI_imCSI}, which indicates that the benefits of increasing $K$ is even more when training larger models for more challenging learning tasks. 
%Our justification is that datasets with more features resulting in gradients with higher empirical variances may highlight the benefits of increasing the number of PS antennas more significantly with the proposed analog FEEL scheme.

\begin{figure}[t!]
\centering
\begin{subfigure}{.5\textwidth}
  \centering
  \includegraphics[scale=0.55,trim={16pt 5pt 43pt 26pt},clip]{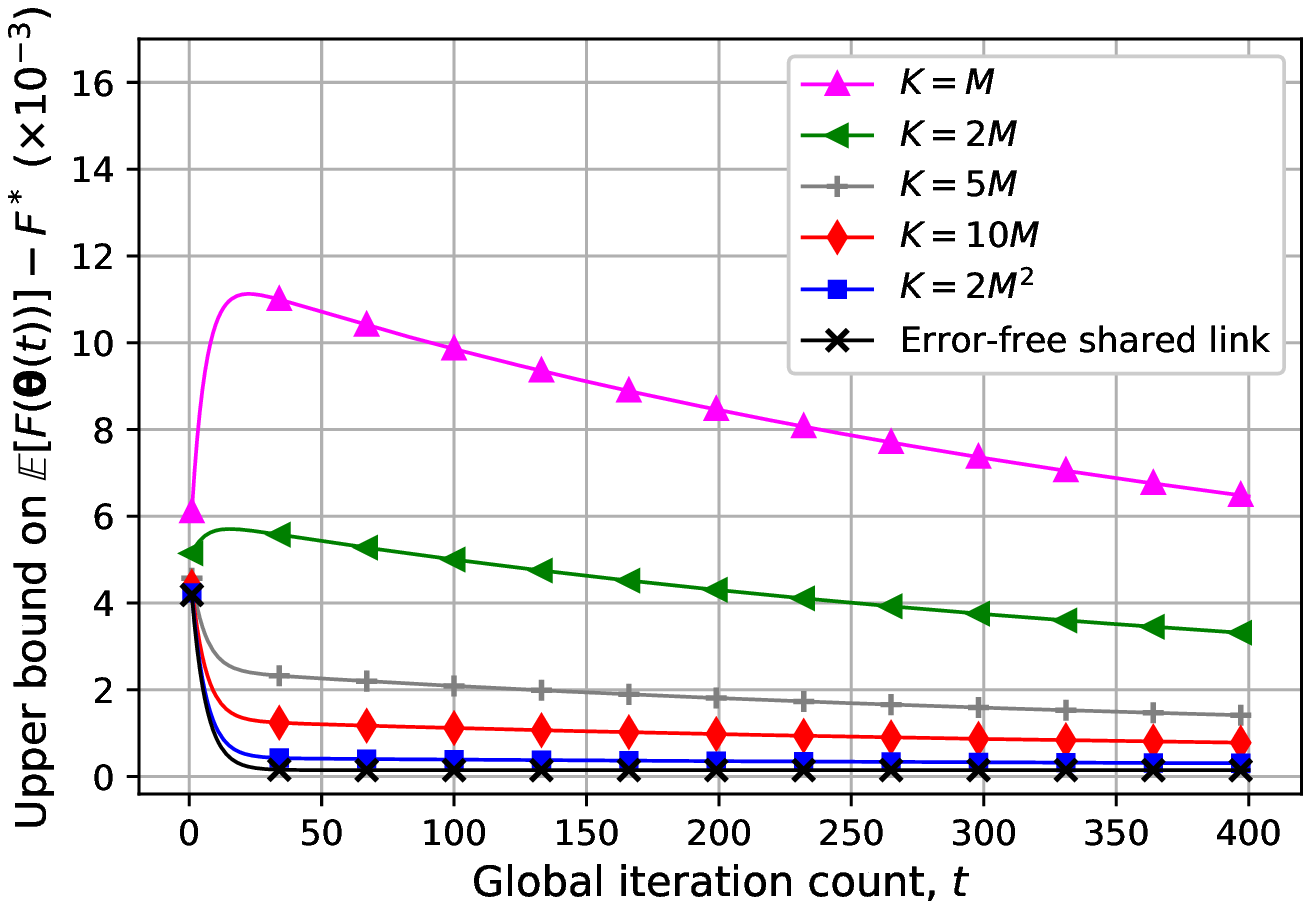}
  \caption{Perfect CSI at PS, $\tilde{\sigma}_h^2 = 0$}
  \label{Fig_conv_IID_perfectCSI}
\end{subfigure}%
\begin{subfigure}{.5\textwidth}
  \centering
  \includegraphics[scale=0.55,trim={16pt 5pt 43pt 26pt},clip]{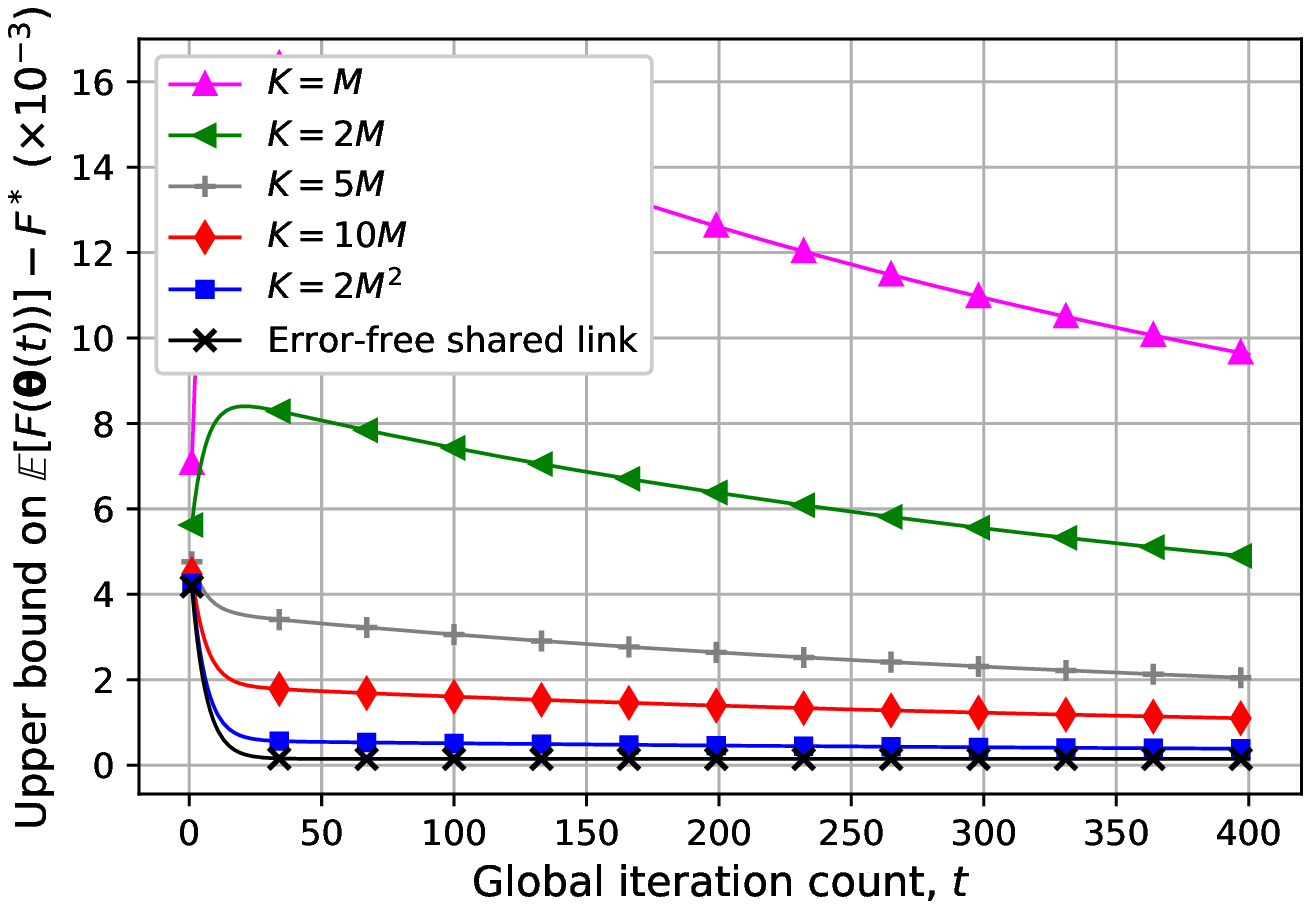}
  \caption{Imperfect CSI at PS, $\tilde{\sigma}_h^2 = M \sigma_h^2/2$}
  \label{Fig_conv_IID_imperfectCSI}
\end{subfigure}
\caption{Upper bound on $\mathbb{E} \left[ F( \boldsymbol{\theta} (T)) \right] - F^*$, given in \eqref{A_ConvF_FstarKequalM}, for different number of antennas $K \in \{ M,2M, 5M, 10M, 2M^2 \}$ with $d = 307498$ parameters used on CIFAR-10 training, $M=20$, $\sigma_z^2 = \sigma_h^2 = 1$, $\tau = 5$, $\mu = 1$, $L=5$, $G^2=\Gamma=1$, $\left\| \boldsymbol{\theta}(0) - \boldsymbol{\theta}^* \right\|_2^2 = 10^3$, and $\eta (t) = \frac{1}{\mu \tau (10^{-4}t + 1)}$.}
\label{Fig_conv_IID}
\end{figure}

In Fig. \ref{Fig_conv_IID}, we illustrate the convergence rate of the proposed analog FEEL algorithm, presented in Corollary \ref{CorollLossGap}, for the setting considered in Fig. \ref{Fig_IID_imCSI}, i.e., training on CIFAR-10 with iid distributed local datasets, for $K \in \{ M,2M,5M, 10M, 2M^2 \}$.
The CNN for training on CIFAR-10, whose architecture is provided in Table \ref{TableCNNArchit}, has $d = 307498$ parameters, and we have $M=20$ and $\sigma_z^2 = \sigma_h^2 = 1$. 
We set $\mu = 1$, $L=5$, $G^2=\Gamma=1$, $\left\| \boldsymbol{\theta}(0) - \boldsymbol{\theta}^* \right\|_2^2 = 10^3$. 
We consider a decreasing learning rate $\eta (t) = \frac{1}{\mu \tau (10^{-4}t + 1)}$, and $\alpha_t = 1 + 10^{-3}t$, $t \in [T]$.
We also consider the convergence rate of the error-free shared link scenario given as follows:
\begin{subequations}
\begin{align}
\mathbb{E} \left[ F( \boldsymbol{\theta} (T)) \right] - F^* \le \frac{L}{2} \Big( \prod\limits_{i=0}^{T-1} A_{\rm{ef}}(i) \Big) \left\| {\boldsymbol{\theta}} (0) - {\boldsymbol{\theta}}^* \right\|_2^2 + \frac{L}{2} \sum\limits_{j=0}^{T-1} B_{\rm{ef}}(j) \prod\limits_{i=j+1}^{T-1} A_{\rm{ef}} (i),
\end{align}
where $0 < \eta(t) \le \min \big\{ 1, \frac{1}{\mu \tau} \big\}$, and we have
\begin{align}
A_{\rm{ef}} (i) \triangleq & 1 - \mu \eta (i) \left( \tau - \eta(i) (\tau - 1 ) \right),\\ 
B_{\rm{ef}}(i) \triangleq & \left( 1+ \mu (1- \eta(i)) \right) \eta^2(i) G^2 \frac{\tau (\tau-1)(2\tau-1)}{6}  + (\tau^2 + \tau-1) \eta^2(i) G^2  + 2  \eta(i) (\tau - 1) \Gamma,
\end{align}
\end{subequations}
which can be obtained by following the procedure presented in the proof of Theorem \ref{A_Theoremtheta_thetastar}. 
We investigate the convergence rate for the cases of having perfect and imperfect CSI at the PS in Fig. \ref{Fig_conv_IID_perfectCSI} and Fig. \ref{Fig_conv_IID_imperfectCSI}, respectively.
We observe that the analytical results illustrated in Fig. \ref{Fig_conv_IID} corroborate the experimental ones presented above, and the theoretical bound obtained for convex loss functions without CSIT approaches that of the perfect communication benchmark with the increasing number of PS antennas.

\section{Conclusions}\label{SecConc}
We have studied FEEL, where colocated wireless devices collaboratively train a global model using their local datasets, and transmit their local updates to the PS over a wireless fading MAC. 
%To make the model more realistic, we have assumed that the devices do not have CSI, and the PS has imperfect CSI about the MAC. 
With the goal of recovering the average local model updates at the PS through over-the-air computation, we have considered analog transmission of the local updates from the devices over the wireless MAC.
The current literature on over-the-air FEEL relies on perfect CSI both at the devices and the PS.
However, acquiring perfect CSI in mobile wireless networks is typically not possible, and even imperfect CSI estimation can introduce delays and waste channel resources.
Therefore, in this work, we have studies FEEL without any CSIT at the devices and with imperfect CSI at the PS.
To mitigate the effects of the time-varying channel without CSI, we assumed that the PS is equipped with multiple antennas, and designed a beamforming technique at the PS to estimate the computation result. 
We have derived the convergence rate of the proposed analog FEEL algorithm that highlights the impact of various system parameters on the performance.
Experimental results on MNIST and CIFAR-10 datasets corroborated the theoretical convergence results, and illustrated that, with the proposed algorithm, increasing the number of PS antennas provides a better estimate of the average local model updates thanks to a better alignment of the desired signals, as well as the elimination of the interference and noise terms. 
Asymptotically, the proposed scheme guarantees that the wireless MAC becomes deterministic, despite the lack of CSIT and perfect CSI at the PS.   
% The PS updates the model parameter based on the signal received over the wireless MAC, and shares it with the workers in a lossless fashion.

\appendices

\section{Proof of Theorem \ref{A_Theoremtheta_thetastar}}\label{A_AppTheorem}

We define an auxiliary variable $\boldsymbol{\upsilon} (t)$ given by 
\begin{align}\label{DefinitionVt}
\boldsymbol{\upsilon} (t+1) \triangleq \boldsymbol{\theta} (t) + \Delta {\boldsymbol{\theta}} (t),
\end{align}
where $\Delta {\boldsymbol{\theta}} (t)$ is as defined in \eqref{DefDeltaThetat}. 
We note that
\begin{align}
\boldsymbol{\theta} (t+1) = \boldsymbol{\theta} (t) + \Delta \wh{\boldsymbol{\theta}} (t).
\end{align}
We have 
\begin{align}\label{AppFDP_1}
& \left\| \boldsymbol{\theta} (t+1) - {\boldsymbol{\theta}}^* \right\|_2^2 =  \left\| \boldsymbol{\theta} (t+1) - {\boldsymbol{\upsilon}} (t+1) + {\boldsymbol{\upsilon}} (t+1)  - {\boldsymbol{\theta}}^* \right\|_2^2 \nonumber \\
& \; = \left\| \boldsymbol{\theta} (t+1) - {\boldsymbol{\upsilon}} (t+1) \right\|_2^2 + \left\| {\boldsymbol{\upsilon}} (t+1)  - {\boldsymbol{\theta}}^* \right\|_2^2  + 2 \langle \boldsymbol{\theta} (t+1) - {\boldsymbol{\upsilon}} (t+1) , {\boldsymbol{\upsilon}} (t+1)  - {\boldsymbol{\theta}}^* \rangle.  
\end{align}
In the following, we bound the three terms on the RHS of \eqref{AppFDP_1}.

\begin{lemma}\label{AppLemmaTerm_1}
We have
\begin{align}\label{AppLemmaTerm_1_Eq_1}
\mathbb{E} \left[ \left\| \boldsymbol{\theta} (t+1) - {\boldsymbol{\upsilon}} (t+1) \right\|_2^2 \right] \le \frac{\big(1+\frac{\tilde{\sigma}_h^2}{M{\sigma}_h^2}\big) \eta^2(t) \tau^2 G^2}{K} + \frac{\big(1+\frac{\tilde{\sigma}_h^2}{M{\sigma}_h^2}\big)\sigma_z^2d}{2 \alpha_t^2 K M \sigma_h^2}.     
\end{align}
\end{lemma}

\begin{proof}
See Appendix \ref{AppProofPDPLemmaTerm_1}. 
\end{proof}

\begin{lemma}\label{AppFDPLemmaTerm_2}
We have
\begin{align}\label{AppFDPBoundTerm_2}
& \mathbb{E} \left[ \left\| \boldsymbol{\upsilon} (t+1) - {\boldsymbol{\theta}}^* \right\|_2^2 \right] \le \left( 1 - \mu \eta (t) \left( \tau - \eta(t) (\tau - 1) \right) \right) \mathbb{E} \left[ \left\| \boldsymbol{\theta} (t) - {\boldsymbol{\theta}}^* \right\|_2^2 \right] \nonumber\\
& \;\;\,+ \left( 1+ \mu (1- \eta(t)) \right) \eta^2(t) G^2 \frac{\tau (\tau-1)(2\tau-1)}{6} + \eta^2(t) (\tau^2 + \tau-1) G^2  + 2 \eta(t) (\tau - 1) \Gamma.      
\end{align}
\end{lemma}
\begin{proof}
See Appendix \ref{AppProofFDPLemmaTerm_2}. 
\end{proof}

\begin{lemma}\label{LemmaThisTermEqualZero}
We have
\begin{align}
\mathbb{E} \big[ \langle \boldsymbol{\theta} (t+1) - {\boldsymbol{\upsilon}} (t+1) , {\boldsymbol{\upsilon}} (t+1)  - {\boldsymbol{\theta}}^* \rangle \big] = 0.    
\end{align}
\begin{proof}
From the definition of $\boldsymbol{\upsilon} (t+1)$ given in \eqref{DefinitionVt}, it follows that
\begin{align}
\mathbb{E} \big[ \langle \boldsymbol{\theta} (t+1) - {\boldsymbol{\upsilon}} (t+1) , {\boldsymbol{\upsilon}} (t+1)  - {\boldsymbol{\theta}}^* \rangle \big] = \mathbb{E} \big[ \langle \Delta \wh{\boldsymbol{\theta}} (t) - \Delta {\boldsymbol{\theta}} (t) , {\boldsymbol{\theta}} (t) + \Delta {\boldsymbol{\theta}} (t)  - {\boldsymbol{\theta}}^* \rangle \big].    
\end{align}
From the independence of $h_{m,k,i} (t)$, $\tilde{h}_{m,k,i} (t)$, and $z_{k,i} (t)$, $\forall m \in [M], \forall k\in [K],\forall i\in [d]$, and \eqref{ReceivedSignalImperfectCSI} and \eqref{DeltaThetaHatImperfectCSI}, expectation of $\Delta \wh{\boldsymbol{\theta}} (t)$ with respect to the channel gains and noise terms results in $\mathbb{E} \big[ \Delta \wh{\boldsymbol{\theta}} (t) \big] = \Delta {\boldsymbol{\theta}} (t)$.
Since the local model updates at the global iteration $t$ are independent of the channel characterizations during the same global iteration, it follows that 
\begin{align}
\mathbb{E} \big[ \langle \Delta \wh{\boldsymbol{\theta}} (t) - \Delta {\boldsymbol{\theta}} (t) , {\boldsymbol{\theta}} (t) + \Delta {\boldsymbol{\theta}} (t)  - {\boldsymbol{\theta}}^* \rangle \big] = 0.    
\end{align}
This completes the proof of Lemma \ref{LemmaThisTermEqualZero}. 
\end{proof}
\end{lemma}

Substituting the results in Lemmas \ref{AppLemmaTerm_1}-\ref{LemmaThisTermEqualZero} into \eqref{AppFDP_1} yields
\begin{align}
&\left\| \boldsymbol{\theta} (t+1) - {\boldsymbol{\theta}}^* \right\|_2^2 \le \left( 1 - \mu \eta (t) \left( \tau - \eta(t) (\tau - 1) \right) \right) \mathbb{E} \left[ \left\| \boldsymbol{\theta} (t) - {\boldsymbol{\theta}}^* \right\|_2^2 \right] \nonumber\\
& \;\;\,+ \frac{\big(1+\frac{\tilde{\sigma}_h^2}{M{\sigma}_h^2}\big) \eta^2(t) \tau^2 G^2}{K} + \frac{\big(1+\frac{\tilde{\sigma}_h^2}{M{\sigma}_h^2}\big)\sigma_z^2d}{2 \alpha_t^2 KM \sigma_h^2} + \left( 1+ \mu (1- \eta(t)) \right) \eta^2(t) G^2 \frac{\tau (\tau-1)(2\tau-1)}{6} \nonumber\\
& \;\;\, + \eta^2(t) (\tau^2 + \tau-1) G^2  + 2 \eta(t) (\tau - 1) \Gamma.
\end{align}
Solving the above inequality recursively concludes Theorem \ref{A_Theoremtheta_thetastar}.

%%%%%%%%%%%%%% Proofs of Lemmas %%%%%%%%%%%%%% 
\section{Proof of Lemma \ref{AppLemmaTerm_1}}\label{AppProofPDPLemmaTerm_1}
We have
\begin{align}\label{APP_sum_of_5_terms}
\mathbb{E} \left[ \left\| \boldsymbol{\theta} (t+1) - {\boldsymbol{\upsilon}} (t+1) \right\|_2^2 \right] =  \mathbb{E} \left[ \big\| \Delta \wh{\boldsymbol{\theta}} (t) - \Delta {\boldsymbol{\theta}} (t) \big\|_2^2 \right] = \sum\limits_{i=1}^{d} \mathbb{E} \left[ \big( \Delta \hat{{\theta}}_i (t) - \Delta {{\theta}}_i (t) \big)^2 \right],   
\end{align}
where $\Delta {{\theta}}_i (t)$ denotes the $i$-th entry of vector $\Delta {\boldsymbol{\theta}} (t)$, for $i \in [d]$. 
In the following, we bound $\mathbb{E} \left[ \big( \Delta \hat{{\theta}}_i (t) - \Delta {{\theta}}_i (t) \big)^2 \right]$, $\forall i$. 
Here we remind that $\Delta \hat{\theta}_{i} \left( t \right) = \sum\nolimits_{l=1}^{5} \Delta \hat{\theta}_{i, l} \left( t \right)$, where $\Delta \hat{\theta}_{i, l} \left( t \right)$ is defined in \eqref{DeltaHatTheta_l}. 
From the independence of $h_{m,k,i} (t)$, $\tilde{h}_{k,i} (t)$, and $z_{k,i} (t)$, $\forall m \in [M], \forall k\in [K],\forall i\in [d]$, and the fact that the local model updates at the global iteration $t$ are independent of the channel realizations during the same global iteration, it is easy to verify that  
\begin{align}
\mathbb{E} \left[ \big( \Delta \hat{{\theta}}_i (t) - \Delta {{\theta}}_i (t) \big)^2 \right] = \mathbb{E} \left[ \big( \Delta \hat{{\theta}}_{i,1} (t) - \Delta {{\theta}}_i (t) \big)^2 \right] + \sum\limits_{l=2}^{5} \mathbb{E} \left[ \Delta \hat{{\theta}}_{i,l}^2 (t) \right].
\end{align}

\begin{lemma}\label{LemmaAppThetaHat_term1}
We have
\begin{align}
\sum\limits_{i=1}^{d} \mathbb{E} \left[ \big( \Delta \hat{{\theta}}_{i,1} (t) - \Delta {{\theta}}_i (t) \big)^2 \right] = \frac{1}{KM^2}  \sum\limits_{m=1}^{M} \mathbb{E} \big[ \left\| \Delta \boldsymbol{\theta}_{m} (t) \right\|_2^2 \big].    
\end{align}
\end{lemma}

\begin{proof}
According to the definition of $\Delta \hat{{\theta}}_{i,1} (t)$, given in \eqref{DeltaHatTheta_l}, we have
\newcommand\firstequal{\mathrel{\overset{\makebox[0pt]{\mbox{\normalfont\tiny\sffamily (a)}}}{=}}}
\newcommand\secondequal{\mathrel{\overset{\makebox[0pt]{\mbox{\normalfont\tiny\sffamily (b)}}}{=}}}
\begin{align}\label{App_Proof_first_term_1}
&\mathbb{E} \left[ \big( \Delta \hat{{\theta}}_{i,1} (t) - \Delta {{\theta}}_i (t) \big)^2 \right] = \mathbb{E} \Big[ \Big( \frac{1}{M} \sum\limits_{m=1}^{M} \Big( \frac{1}{K \sigma^2_h} \sum\limits_{k=1}^{K} \left| {h}_{m,k,i} (t) \right|^2 -1 \Big)  \Delta {\theta}_{m,i} (t) \Big)^2 \Big] \nonumber\\
& = \mathbb{E} \Big[ \frac{1}{M^2} \sum\limits_{m_1=1}^{M} \sum\limits_{m_2=1}^{M} \Big(  1 - \frac{1}{K \sigma^2_h} \sum\limits_{k=1}^{K} \left| {h}_{m_1,k,i} (t) \right|^2 - \frac{1}{K \sigma^2_h} \sum\limits_{k=1}^{K} \left| {h}_{m_2,k,i} (t) \right|^2 \Big.\nonumber\\
& \quad \quad \; \Big. + \frac{1}{K^2 \sigma^4_h} \sum\limits_{k_1=1}^{K} \sum\limits_{k_2=1}^{K} \left| {h}_{m_1,k_1,i} (t) \right|^2 \left| {h}_{m_2,k_2,i} (t) \right|^2 \Big) \Delta {\theta}_{m_1,i} (t) \Delta {\theta}_{m_2,i} (t) \Big] \nonumber\\
& \firstequal \mathbb{E} \Big[ \frac{1}{M^2} \sum\limits_{m=1}^{M} \Big( -\frac{1}{K} + \frac{1}{K^2 \sigma_h^4} \sum\limits_{k=1}^{K} \left| {h}_{m,k,i} (t) \right|^4 \Big) \Delta {\theta}^2_{m,i} (t) \Big] \secondequal \mathbb{E} \Big[ \frac{1}{KM^2} \sum\limits_{m=1}^{M} \Delta {\theta}^2_{m,i} (t) \Big], 
\end{align}
where (a) follows from the independence of $h_{m,k,i} (t)$, $\forall m, k$, and (b) follows since $\mathbb{E} \big[ \left| {h}_{m,k,i} (t) \right|^2 \big] = \sigma^2_h$ and $\mathbb{E} \big[ \left| {h}_{m,k,i} (t) \right|^4 \big] = 2 \sigma^4_h$. 
Lemma \ref{LemmaAppThetaHat_term1} follows from \eqref{App_Proof_first_term_1}.
\end{proof}

\begin{lemma}\label{LemmaAppThetaHat_term2}
We have
\begin{align}
\sum\limits_{i=1}^{d} \mathbb{E} \left[ \Delta \hat{{\theta}}^2_{i,2} (t) \right] = \frac{(M-1)}{KM^2}  \sum\limits_{m=1}^{M} \mathbb{E} \big[ \left\| \Delta \boldsymbol{\theta}_{m} (t) \right\|_2^2 \big].
\end{align}
\end{lemma}

\begin{proof}
We first consider $1 \le i \le d/2$. By substituting $\Delta \hat{{\theta}}_{i,2} (t)$ from \eqref{DeltaHatTheta_l}, it follows that
\newcommand\firstequal{\mathrel{\overset{\makebox[0pt]{\mbox{\normalfont\tiny\sffamily (a)}}}{=}}}
\newcommand\secondequal{\mathrel{\overset{\makebox[0pt]{\mbox{\normalfont\tiny\sffamily (b)}}}{=}}}
\newcommand\thirdequal{\mathrel{\overset{\makebox[0pt]{\mbox{\normalfont\tiny\sffamily (c)}}}{=}}}
\begin{align}\label{App_Proof_Second_term_1}
&\mathbb{E} \left[ \Delta \hat{{\theta}}^2_{i,2} (t) \right]  = \mathbb{E} \Big[ \Big( \frac{1}{KM \sigma_h^2} \sum\limits_{m=1}^{M} \sum\limits_{m'=1, m' \ne m}^{M}  \sum\limits_{k=1}^{K} {\rm{Re}} \left\{ \left( {h}_{m,k,i} (t) \right)^{*} {h}_{m',k, i} (t) \left( \Delta {\theta}_{m',i} (t) + j \Delta {\theta}_{m',d/2+i} (t) \right) \right\} \Big)^2 \Big] \nonumber\\
& \firstequal \mathbb{E} \Big[ \frac{1}{K^2 M^2  \sigma_h^4} \sum\limits_{m=1}^{M} \sum\limits_{m'=1, m' \ne m}^{M}  \sum\limits_{k=1}^{K} \Big( \left( {\rm{Re}} \left\{ \left( {h}_{m,k,i} (t) \right)^{*} {h}_{m',k, i} (t) \left( \Delta {\theta}_{m',i} (t) + j \Delta {\theta}_{m',d/2+i} (t) \right) \right\} \right)^2 \Big. \nonumber\\
&\quad \;+ {\rm{Re}} \left\{ \left( {h}_{m,k,i} (t) \right)^{*} {h}_{m',k, i} (t) \left( \Delta {\theta}_{m',i} (t) + j \Delta {\theta}_{m',d/2+i} (t) \right) \right\} \nonumber\\
&\qquad \;\,\, \Big. {\rm{Re}} \left\{ \left( {h}_{m',k,i} (t) \right)^{*} {h}_{m,k, i} (t) \left( \Delta {\theta}_{m,i} (t) + j \Delta {\theta}_{m,d/2+i} (t) \right) \right\} \Big) \Big] \nonumber
\end{align}
\begin{align}
& \secondequal \mathbb{E} \Big[ \frac{1}{2KM^2} \sum\limits_{m=1}^{M} \Big( (M-1) \left( \Delta {\theta}^2_{m,i} (t) + \Delta {\theta}_{m,d/2+i}^2 (t) \right) \nonumber\\
& \qquad \; \qquad \qquad \quad \; + \sum\limits_{m'=1, m'\ne m}^{M} \left( \Delta {\theta}_{m,i} (t) \Delta {\theta}_{m',i} (t) - \Delta {\theta}_{m,d/2+i} (t) \Delta {\theta}_{m',d/2+i} (t)  \right) \Big) \Big],
\end{align}
where (a) and (b) follow from the independence of $h_{m,k,i} (t)$, $\forall m, k$, and $\mathbb{E} \big[ \left| {h}_{m,k,i} (t) \right|^2 \big] = \sigma^2_h$, respectively. 
Similarly, for $d/2+1 \le i \le d$, if follows that
\begin{align}\label{App_Proof_Second_term_2}
&\mathbb{E} \left[ \Delta \hat{{\theta}}^2_{i,2} (t) \right]  = \mathbb{E} \Big[ \Big( \frac{1}{KM \sigma_h^2} \sum\limits_{m=1}^{M} \sum\limits_{m'=1, m' \ne m}^{M}  \sum\limits_{k=1}^{K} {\rm{Im}} \left\{ \left( {h}_{m,k,i} (t) \right)^{*} {h}_{m',k, i} (t) \left( \Delta {\theta}_{m',i} (t) + j \Delta {\theta}_{m',d/2+i} (t) \right) \right\} \Big)^2 \Big] \nonumber\\
& = \mathbb{E} \Big[ \frac{1}{2KM^2} \sum\limits_{m=1}^{M} \Big( (M-1) \left( \Delta {\theta}^2_{m,i-d/2} (t) + \Delta {\theta}_{m,i}^2 (t) \right) \nonumber\\
& \qquad \; \qquad \qquad \quad \; + \sum\limits_{m'=1, m'\ne m}^{M} \left( \Delta {\theta}_{m,i} (t) \Delta {\theta}_{m',i} (t) - \Delta {\theta}_{m,i-d/2} (t) \Delta {\theta}_{m',i-d/2} (t)  \right) \Big) \Big].
\end{align}\end{proof}
From \eqref{App_Proof_Second_term_1} and \eqref{App_Proof_Second_term_2}, it follows that 
\begin{align}
\sum\limits_{i=1}^{d} \mathbb{E} \left[ \Delta \hat{{\theta}}^2_{i,2} (t) \right] &= \mathbb{E} \Big[ \frac{(M-1)}{KM^2}  \sum\limits_{m=1}^{M} \sum\limits_{i=1}^{d/2} \left( \Delta {\theta}^2_{m,i} (t) + \Delta {\theta}_{m,d/2+i}^2 (t) \right) \Big] \nonumber\\
&= \frac{(M-1)}{KM^2}  \sum\limits_{m=1}^{M} \mathbb{E} \big[ \left\| \Delta \boldsymbol{\theta}_{m} (t) \right\|_2^2 \big].    
\end{align}

\begin{lemma}\label{LemmaAppThetaHat_term3}
We have
\begin{align}\label{EQ_Lemma_App_thirs_term}
\sum\limits_{i=1}^{d} \mathbb{E} \left[ \Delta \hat{{\theta}}^2_{i,3} (t) \right] = \frac{\sigma_z^2 d}{2 \alpha_t^2 KM  \sigma_h^2}.
\end{align}
\end{lemma}

\begin{proof}
According to the definition of $\Delta \hat{{\theta}}_{i,3} (t)$, given in \eqref{DeltaHatTheta_l}, for $1 \le i \le d/2$ we have
\newcommand\firstequal{\mathrel{\overset{\makebox[0pt]{\mbox{\normalfont\tiny\sffamily (a)}}}{=}}}
\newcommand\secondequal{\mathrel{\overset{\makebox[0pt]{\mbox{\normalfont\tiny\sffamily (b)}}}{=}}}
\begin{align}\label{EQ_Proof_App_thirs_term}
\mathbb{E} & \left[ \Delta \hat{{\theta}}^2_{i,3} (t) \right] = \mathbb{E} \Big[ \Big( \frac{1}{\alpha_t K M \sigma_h^2} \sum\limits_{m=1}^{M} \sum\limits_{k=1}^{K} {\rm{Re}} \left\{ \left( {h}_{m,k,i} (t) \right)^{*} {z}_{k, i} (t)  \right\} \Big)^2 \Big] \nonumber\\
& \quad \quad \quad \quad \, \firstequal \mathbb{E} \Big[  \frac{1}{\alpha_t^2 K^2 M^2  \sigma_h^4} \sum\limits_{m=1}^{M} \sum\limits_{k=1}^{K} \left( {\rm{Re}} \left\{ \left( {h}_{m,k,i} (t) \right)^{*} {z}_{k, i} (t)  \right\} \right)^2 \Big] \secondequal \frac{\sigma_z^2}{2 \alpha_t^2 K M \sigma_h^2},
\end{align}
where (a) follows from the independence of $h_{m,k,i} (t)$ and $z_{k,i} (t)$, $\forall m, k$, and (b) follows since $\mathbb{E} \big[ \left| {h}_{m,k,i} (t) \right|^2 \big] = \sigma^2_h$ and $\mathbb{E} \big[ \left| {z}_{k,i} (t) \right|^2 \big] = \sigma^2_z$. 
The same result can be obtained for $d/2+1 \le i \le d$ by following the same procedure as above.
It is straightforward to derive \eqref{EQ_Lemma_App_thirs_term} from \eqref{EQ_Proof_App_thirs_term}.  
\end{proof}

\begin{lemma}\label{LemmaAppThetaHat_term4}
We have
\begin{align}
\sum\limits_{i=1}^{d} \mathbb{E} \left[ \Delta \hat{{\theta}}^2_{i,4} (t) \right] =
\frac{\tilde{\sigma}_h^2}{K M^2\sigma_h^2} \sum\limits_{m=1}^{M} \mathbb{E} \big[ \left\| \Delta \boldsymbol{\theta}_m (t) \right\|_2^2 \big].
\end{align}
\end{lemma}

\begin{proof}
By Substituting $\Delta \hat{{\theta}}_{i,4} (t)$ from \eqref{DeltaHatTheta_l}, for $1 \le i \le d/2$, we have
\newcommand\firstequal{\mathrel{\overset{\makebox[0pt]{\mbox{\normalfont\tiny\sffamily (a)}}}{=}}}
\newcommand\secondequal{\mathrel{\overset{\makebox[0pt]{\mbox{\normalfont\tiny\sffamily (b)}}}{=}}}
\begin{align}\label{App_Proof_Fourth_term_1}
& \mathbb{E} \left[ \Delta \hat{{\theta}}^2_{i,4} (t) \right]  = \mathbb{E} \Big[ \Big( \frac{1}{KM \sigma_h^2} \sum\limits_{m=1}^{M} \sum\limits_{k=1}^{K} {\rm{Re}} \big\{ \big( \tilde{h}_{k,i} (t) \big)^{*} {h}_{m,k, i} (t) \left( \Delta {\theta}_{m,i} (t) + j \Delta {\theta}_{m,d/2+i} (t) \right) \big\} \Big)^2 \Big] \nonumber\\
& \firstequal \mathbb{E} \Big[  \frac{1}{K^2 M^2 \sigma_h^4} \sum\limits_{m=1}^{M} \sum\limits_{k=1}^{K} \left( {\rm{Re}} \big\{ \big( \tilde{h}_{k,i} (t) \big)^{*} {h}_{m,k, i} (t) \left( \Delta {\theta}_{m,i} (t) + j \Delta {\theta}_{m,d/2+i} (t) \right)  \big\} \right)^2 \Big] \nonumber\\
& \secondequal \mathbb{E} \Big[ \frac{\tilde{\sigma}_h^2}{2 K M^2 \sigma_h^2} \sum\limits_{m=1}^{M} \left( \Delta {\theta}^2_{m,i} (t) + \Delta {\theta}_{m,d/2+i}^2 (t) \right) \Big],
\end{align}
where (a) follows from the independence of $\tilde{h}_{k,i} (t)$ and ${h}_{m,k,i} (t)$, $\forall m, k$, and (b) is the result of $\mathbb{E} \big[ \big| \tilde{h}_{k,i} (t) \big|^2 \big] = \tilde{\sigma}^2_h$ and $\mathbb{E} \big[ \left| {h}_{m,k,i} (t) \right|^2 \big] = \sigma^2_h$, $\forall i$.
Similarly, for $d/2+1 \le i \le d$, we can obtain
\begin{align}\label{App_Proof_Fourth_term_2}
& \mathbb{E} \left[ \Delta \hat{{\theta}}^2_{i,4} (t) \right] = \mathbb{E} \Big[ \frac{\tilde{\sigma}_h^2}{2 K M^2 \sigma_h^2} \sum\limits_{m=1}^{M} \left( \Delta {\theta}^2_{m,i-d/2} (t) + \Delta {\theta}_{m,d/2}^2 (t) \right) \Big].
\end{align}
From \eqref{App_Proof_Fourth_term_1} and \eqref{App_Proof_Fourth_term_2}, we have
\begin{align}
\sum\limits_{i=1}^{d} \mathbb{E} \left[ \Delta \hat{{\theta}}^2_{i,4} (t) \right] &= \mathbb{E} \Big[ \frac{\tilde{\sigma}_h^2}{K M^2 {\sigma}_h^2}  \sum\limits_{m=1}^{M} \sum\limits_{i=1}^{d/2} \left( \Delta {\theta}^2_{m,i} (t) + \Delta {\theta}_{m,d/2+i}^2 (t) \right) \Big]  = \frac{\tilde{\sigma}_h^2}{K M^2\sigma_h^2} \sum\limits_{m=1}^{M} \mathbb{E} \big[ \left\| \Delta \boldsymbol{\theta}_m (t) \right\|_2^2 \big].   
\end{align}
\end{proof}

\begin{lemma}\label{LemmaAppThetaHat_term5}
We have
\begin{align}
\sum\limits_{i=1}^{d} \mathbb{E} \left[ \Delta \hat{{\theta}}^2_{i,5} (t) \right] = \frac{\tilde{\sigma}_h^2 \sigma_z^2 d}{2 \alpha_t^2 K M^2 \sigma_h^4}.
\end{align}
\end{lemma}

\begin{proof}
From the definition of $\Delta \hat{{\theta}}_{i,5} (t)$, given in \eqref{DeltaHatTheta_l}, for $1 \le i \le d/2$ we have
\newcommand\firstequal{\mathrel{\overset{\makebox[0pt]{\mbox{\normalfont\tiny\sffamily (a)}}}{=}}}
\newcommand\secondequal{\mathrel{\overset{\makebox[0pt]{\mbox{\normalfont\tiny\sffamily (b)}}}{=}}}
\begin{align}\label{App_Proof_Fifth_term_1}
\mathbb{E} & \left[ \Delta \hat{{\theta}}^2_{i,5} (t) \right] = \mathbb{E} \Big[ \Big( \frac{1}{\alpha_t K M \sigma_h^2} \sum\limits_{k=1}^{K} {\rm{Re}} \big\{ \big( \tilde{h}_{k,i} (t) \big)^{*} {z}_{k, i} (t)  \big\} \Big)^2 \Big] \nonumber\\
& \quad \quad \quad \quad \, \firstequal \mathbb{E} \Big[  \frac{1}{\alpha_t^2 K^2 M^2 \sigma_h^4} \sum\limits_{k=1}^{K} \big( {\rm{Re}} \big\{ \big( \tilde{h}_{k,i} (t) \big)^{*} {z}_{k, i} (t)  \big\} \big)^2 \Big] \secondequal \frac{\tilde{\sigma}_h^2 \sigma_z^2}{2 \alpha_t^2 K M^2 \sigma_h^4},
\end{align}
where (a) follows from the independence of $\tilde{h}_{k,i} (t)$ and $z_{k,i} (t)$, $\forall k$, and (b) follows since $\mathbb{E} \big[ \big| \tilde{h}_{k,i} (t) \big|^2 \big] = \tilde{\sigma}^2_h$ and $\mathbb{E} \big[ \left| {z}_{k,i} (t) \right|^2 \big] = \sigma^2_z$. 
The same result can be obtained for $d/2+1 \le i \le d$ by following the same procedure as above.
The proof of Lemma \ref{LemmaAppThetaHat_term5} is completed from \eqref{App_Proof_Fifth_term_1}. 
\end{proof}
By substituting the results of Lemmas \ref{LemmaAppThetaHat_term1}-\ref{LemmaAppThetaHat_term5} into \eqref{APP_sum_of_5_terms}, it follows that 
\newcommand\firstequal{\mathrel{\overset{\makebox[0pt]{\mbox{\normalfont\tiny\sffamily (a)}}}{=}}}
\newcommand\firstinequal{\mathrel{\overset{\makebox[0pt]{\mbox{\normalfont\tiny\sffamily (b)}}}{\le}}}
\newcommand\secondinequal{\mathrel{\overset{\makebox[0pt]{\mbox{\normalfont\tiny\sffamily (c)}}}{\le}}}
\begin{align}
&\mathbb{E} \left[ \left\| \boldsymbol{\theta} (t+1) - {\boldsymbol{\upsilon}} (t+1) \right\|_2^2 \right] = \frac{\big(1+\frac{\tilde{\sigma}_h^2}{M{\sigma}_h^2}\big)}{KM} \sum\limits_{m=1}^{M} \mathbb{E} \big[ \left\| \Delta \boldsymbol{\theta}_m (t) \right\|_2^2 \big] + \frac{\big(1+\frac{\tilde{\sigma}_h^2}{M{\sigma}_h^2}\big)\sigma_z^2d}{2 \alpha_t^2 K M \sigma_h^2}\nonumber\\
&\firstequal \frac{\Big(1+\frac{\tilde{\sigma}_h^2}{M{\sigma}_h^2}\Big) \eta^2(t)}{KM} \sum\limits_{m=1}^{M} \mathbb{E} \Big[ \Big\| \sum\limits_{l=1}^{\tau} \nabla F_m \left( \boldsymbol{\theta}_m^l (t), \xi_m^l (t) \right) \Big\|_2^2 \Big] + \frac{\big(1+\frac{\tilde{\sigma}_h^2}{M{\sigma}_h^2}\big)\sigma_z^2d}{2 \alpha_t^2 K M \sigma_h^2} \nonumber\\
& \firstinequal \frac{\Big(1+\frac{\tilde{\sigma}_h^2}{M{\sigma}_h^2}\Big) \eta^2(t) \tau}{KM} \sum\limits_{m=1}^{M} \sum\limits_{l=1}^{\tau} \mathbb{E} \Big[ \left\|  \nabla F_m \left( \boldsymbol{\theta}_m^l (t), \xi_m^l (t) \right) \right\|_2^2 \Big] + \frac{\big(1+\frac{\tilde{\sigma}_h^2}{M{\sigma}_h^2}\big)\sigma_z^2d}{2 \alpha_t^2 K M \sigma_h^2}\nonumber\\
& \secondinequal \frac{\Big(1+\frac{\tilde{\sigma}_h^2}{M{\sigma}_h^2}\Big) \eta^2(t) \tau^2 G^2}{K} + \frac{\big(1+\frac{\tilde{\sigma}_h^2}{M{\sigma}_h^2}\big)\sigma_z^2d}{2 \alpha_t^2 K M \sigma_h^2},
\end{align}
where (a) follows by replacing $\Delta \boldsymbol{\theta}_m (t)$ from \eqref{DeltaTheta_m_t_conver}, (b) is due to the convexity of $\left\| \cdot \right\|_2^2$, and (c) follows from Assumption \ref{AssumpBoundedVarGradient}.

\section{Proof of Lemma \ref{AppFDPLemmaTerm_2}}\label{AppProofFDPLemmaTerm_2}

We follow the same procedure as the one used to prove \cite[Lemma 3]{FLConvergenceMohDenSanjVince}. We have
\begin{align}\label{AppLemmaTemr_2_Eq_1}
\mathbb{E} \left[ \left\| \boldsymbol{\upsilon} (t+1) - {\boldsymbol{\theta}}^* \right\|_2^2 \right] & = \mathbb{E} \left[ \left\| \boldsymbol{\theta} (t) + \Delta {\boldsymbol{\theta}} (t) - {\boldsymbol{\theta}}^* \right\|_2^2 \right]  \nonumber\\
& = \mathbb{E} \left[ \left\| \boldsymbol{\theta} (t) - {\boldsymbol{\theta}}^* \right\|_2^2 \right] + \mathbb{E} \left[ \left\| \Delta {\boldsymbol{\theta}} (t) \right\|_2^2 \right] + 2 \mathbb{E} \left[ \langle \boldsymbol{\theta} (t) - {\boldsymbol{\theta}}^* , \Delta {\boldsymbol{\theta}} (t) \rangle \right].
\end{align}
From the convexity of $\left\| \cdot \right\|_2^2$, it follows that
\newcommand\onestequal{\mathrel{\overset{\makebox[0pt]{\mbox{\normalfont\tiny\sffamily (a)}}}{=}}}
\newcommand\binequal{\mathrel{\overset{\makebox[0pt]{\mbox{\normalfont\tiny\sffamily (b)}}}{\le}}}
\begin{align}\label{AppLemmaTemr_2_Eq_1_2}
\mathbb{E} \left[ \left\| \Delta {\boldsymbol{\theta}} (t) \right\|_2^2 \right] & \le  \frac{1}{M} \sum\limits_{m =1}^{M} \mathbb{E} \left[ \left\| \Delta \boldsymbol{\theta}_m (t) \right\|_2^2 \right] \onestequal \frac{\eta^2(t)}{M} \sum\limits_{m =1}^{M} \mathbb{E} \Big[ \Big\| \sum\limits_{i=1}^{\tau} \nabla F_m \left( \boldsymbol{\theta}_m^i (t), \xi_m^i (t) \right) \Big\|_2^2 \Big] \nonumber\\
& \le \frac{\eta^2(t) \tau}{M} \sum\limits_{m =1}^{M} \sum\limits_{i=1}^{\tau} \mathbb{E} \left[ \left\| \nabla F_m \left( \boldsymbol{\theta}_m^i (t), \xi_m^i (t) \right) \right\|_2^2 \right] \binequal \eta^2(t) \tau^2 G^2,
\end{align}
where (a) follows by replacing $\Delta \boldsymbol{\theta}_m (t)$ from \eqref{DeltaTheta_m_t_conver}, and (b) follows from Assumption \ref{AssumpBoundedVarGradient}. 
Plugging the above inequality into \eqref{AppLemmaTemr_2_Eq_1} yields
\begin{align}\label{AppLemmaTemr_2_Eq_1_3}
&\mathbb{E} \left[ \left\| \boldsymbol{\upsilon} (t+1) - {\boldsymbol{\theta}}^* \right\|_2^2 \right] \le \mathbb{E} \left[ \left\| \boldsymbol{\theta} (t) - {\boldsymbol{\theta}}^* \right\|_2^2 \right] + \eta^2(t) \tau^2 G^2 + 2 \mathbb{E} \left[ \langle \boldsymbol{\theta} (t) - {\boldsymbol{\theta}}^* , \Delta {\boldsymbol{\theta}} (t) \rangle \right].
\end{align}
We bound the last term on the RHS of the above inequality.
We have \vspace{-.1cm}
\begin{align}\label{AppLemmaTemr_2_Eq_2}
2 \mathbb{E} \left[ \langle \boldsymbol{\theta} (t) - {\boldsymbol{\theta}}^* , \Delta {\boldsymbol{\theta}} (t) \rangle \right] & \onestequal \frac{2}{M} \sum\limits_{m=1}^{M} \mathbb{E} \left[ \langle \boldsymbol{\theta} (t) - {\boldsymbol{\theta}}^* , \Delta {\boldsymbol{\theta}}_m (t) \rangle \right] \nonumber\\
&= \frac{2 \eta(t)}{M} \sum\limits_{m=1}^{M} \mathbb{E} \Big[ \langle {\boldsymbol{\theta}}^* - \boldsymbol{\theta} (t) , \sum\limits_{i=1}^{\tau} \nabla F_m \left( \boldsymbol{\theta}_m^i (t), \xi_m^i (t) \right) \rangle \Big] \nonumber\\
& = \frac{2 \eta(t)}{M} \sum\limits_{m=1}^{M} \mathbb{E} \left[ \langle {\boldsymbol{\theta}}^* - \boldsymbol{\theta} (t) , \nabla F_m \left( \boldsymbol{\theta} (t), \xi_m^1 (t) \right) \rangle \right] \nonumber \\
& \quad + \frac{2 \eta(t)}{M} \sum\limits_{m=1}^{M} \mathbb{E} \Big[ \langle {\boldsymbol{\theta}}^* - \boldsymbol{\theta} (t) , \sum\limits_{i=2}^{\tau} \nabla F_m \left( \boldsymbol{\theta}_m^i (t), \xi_m^i (t) \right) \rangle \Big]. 
\end{align}
Next we bound the two terms on the RHS of the above equality. 
We have \vspace{-.1cm}
\newcommand\aequal{\mathrel{\overset{\makebox[0pt]{\mbox{\normalfont\tiny\sffamily (a)}}}{=}}}
\newcommand\thirdinequal{\mathrel{\overset{\makebox[0pt]{\mbox{\normalfont\tiny\sffamily (b)}}}{\le}}}
\newcommand\cinequal{\mathrel{\overset{\makebox[0pt]{\mbox{\normalfont\tiny\sffamily (c)}}}{\le}}}
\begin{align}\label{AppLemmaTemr_2_Eq_3}
& \frac{2 \eta(t)}{M} \sum\limits_{m=1}^{M} \mathbb{E} \left[ \langle {\boldsymbol{\theta}}^* - \boldsymbol{\theta} (t) , \nabla F_m \left( \boldsymbol{\theta} (t), \xi_m^1 (t) \right) \rangle \right]   \aequal \frac{2 \eta(t)}{M} \sum\limits_{m=1}^{M} \mathbb{E} \left[ \langle {\boldsymbol{\theta}}^* - \boldsymbol{\theta} (t) , \nabla F_m \left( \boldsymbol{\theta} (t) \right) \rangle \right] \nonumber\\
&  \thirdinequal \frac{2 \eta(t)}{M} \sum\limits_{m=1}^{M} \mathbb{E} \big[ F_m (\boldsymbol{\theta}^*) - F_m(\boldsymbol{\theta} (t)) - \frac{\mu}{2} \left\| \boldsymbol{\theta} (t) - {\boldsymbol{\theta}}^* \right\|_2^2 \big]  = 2 \eta(t) \big( F^* - \mathbb{E} \left[ F(\boldsymbol{\theta} (t)) \right] - \frac{\mu}{2} \mathbb{E} \left[ \left\| \boldsymbol{\theta} (t) - {\boldsymbol{\theta}}^* \right\|_2^2 \right] \big) \nonumber\\
& \cinequal - \mu \eta(t) \mathbb{E} \left[ \left\| \boldsymbol{\theta} (t) - {\boldsymbol{\theta}}^* \right\|_2^2 \right],
\end{align}
where (a) follows since $\mathbb{E}_{\xi} \left[ \nabla F_m \left( \boldsymbol{\theta} (t), \xi^1_m (t) \right) \right] = \nabla F_m \left( \boldsymbol{\theta} (t)  \right)$, (b) follows since $F_m$ is $\mu$-strongly convex, and (c) holds because $F^* \le F(\boldsymbol{\theta} (t))$. 
For the second term on the RHS of \eqref{AppLemmaTemr_2_Eq_2}, we have \vspace{-.9cm} 
\begin{align}\label{AppLemmaTemr_2_Eq_4}
& \frac{2 \eta(t)}{M} \sum\limits_{m=1}^{M} \mathbb{E} \Big[ \langle {\boldsymbol{\theta}}^* - \boldsymbol{\theta} (t) , \sum\limits_{i=2}^{\tau} \nabla F_m \left( \boldsymbol{\theta}_m^i (t), \xi_m^i (t) \right) \rangle \Big]  = \frac{2 \eta(t)}{M} \sum\limits_{m=1}^{M} \sum\limits_{i=2}^{\tau} \mathbb{E} \left[ \langle {\boldsymbol{\theta}}^* - \boldsymbol{\theta} (t) , \nabla F_m \left( \boldsymbol{\theta}_m^i (t), \xi_m^i (t) \right) \rangle \right] \nonumber\\
& = \frac{2 \eta(t)}{M} \sum\limits_{m=1}^{M} \sum\limits_{i=2}^{\tau} \mathbb{E} \left[ \langle \boldsymbol{\theta}_m^i (t) - \boldsymbol{\theta} (t) , \nabla F_m \left( \boldsymbol{\theta}_m^i (t), \xi_m^i (t) \right) \rangle \right] \nonumber\\
& \quad + \frac{2 \eta(t)}{M} \sum\limits_{m=1}^{M} \sum\limits_{i=2}^{\tau} \mathbb{E} \left[ \langle \boldsymbol{\theta}^* - \boldsymbol{\theta}_m^i (t) , \nabla F_m \left( \boldsymbol{\theta}_m^i (t), \xi_m^i (t) \right) \rangle \right].
\end{align}
From Cauchy-Schwarz inequality, it follows that
\newcommand\ainequal{\mathrel{\overset{\makebox[0pt]{\mbox{\normalfont\tiny\sffamily (a)}}}{\le}}}
\begin{align}\label{AppLemmaTemr_2_Eq_5}
& \frac{2 \eta(t)}{M} \sum\limits_{m=1}^{M} \sum\limits_{i=2}^{\tau} \mathbb{E} \left[ \langle \boldsymbol{\theta}_m^i (t) - \boldsymbol{\theta} (t) , \nabla F_m \left( \boldsymbol{\theta}_m^i (t), \xi_m^i (t) \right) \rangle \right] \nonumber\\
& \; \; \; \quad \le \frac{ \eta(t)}{M} \sum\limits_{m=1}^{M} \sum\limits_{i=2}^{\tau} \mathbb{E} \Big[ \frac{1}{\eta(t)} \left\| \boldsymbol{\theta}_m^i (t) - \boldsymbol{\theta} (t) \right\|_2^2 + \eta(t) \left\| \nabla F_m \left( \boldsymbol{\theta}_m^i (t), \xi_m^i (t) \right) \right\|_2^2 \Big]\nonumber\\
& \; \; \; \quad \ainequal \frac{1}{M} \sum\limits_{m=1}^{M} \sum\limits_{i=2}^{\tau} \mathbb{E} \big[ \left\| \boldsymbol{\theta}_m^i (t) - \boldsymbol{\theta} (t) \right\|_2^2 \big] + \eta^2 (t) \left( \tau - 1 \right) G^2, 
\end{align}
where (a) follows from Assumption \ref{AssumpBoundedVarGradient}.  
Also, the following lemma presents an upper bound on the second term in the RHS of \eqref{AppLemmaTemr_2_Eq_4}.

\begin{lemma}\label{LemmaTermE}
The second term on the RHS of \eqref{AppLemmaTemr_2_Eq_4} is upper bounded as follows:
\begin{align}\label{EQ_LemmaTermE}
& \frac{2  \eta(t)}{M} \sum\limits_{m=1}^{M} \sum\limits_{i=2}^{\tau} \mathbb{E} \left[ \langle \boldsymbol{\theta}^* - \boldsymbol{\theta}_m^i (t) , \nabla F_m \left( \boldsymbol{\theta}_m^i (t), \xi_m^i (t) \right) \rangle \right]  \le - \mu  \eta(t) (1 - \eta(t)) (\tau - 1) \mathbb{E} \left[ \left\| \boldsymbol{\theta} (t) - \boldsymbol{\theta}^* \right\|_2^2 \right] \nonumber\\
& \qquad \qquad \quad + \frac{\mu (1-\eta(t))}{M} \sum\limits_{m=1}^{M} \sum\limits_{i=2}^{\tau} \mathbb{E} \left[ \left\| \boldsymbol{\theta}_m^i (t) - {\boldsymbol{\theta}} (t) \right\|_2^2 \right] + 2 \eta(t) (\tau - 1) \Gamma.
\end{align}
\end{lemma}
\begin{proof}
See Appendix \ref{AppProofLemmaTermE}. 
\end{proof}

Substituting the results in \eqref{AppLemmaTemr_2_Eq_5} and \eqref{EQ_LemmaTermE} into \eqref{AppLemmaTemr_2_Eq_4} yields
\begin{align}\label{AppLemmaTemr_2_Eq_6}
& \frac{2 \eta(t)}{M} \sum\limits_{m=1}^{M} \mathbb{E} \Big[ \langle {\boldsymbol{\theta}}^* - \boldsymbol{\theta} (t) , \sum\limits_{i=2}^{\tau} \nabla F_m \left( \boldsymbol{\theta}_m^i (t), \xi_m^i (t) \right) \rangle \Big]   \le - \mu \eta(t) (1 - \eta(t)) (\tau - 1) \mathbb{E} \left[ \left\| \boldsymbol{\theta} (t) - \boldsymbol{\theta}^* \right\|_2^2 \right] \nonumber\\
&\;\;\;\; \quad + \frac{\left( 1+ \mu (1- \eta(t)) \right)}{M} \sum\limits_{m=1}^{M} \sum\limits_{i=2}^{\tau} \mathbb{E} \big[ \left\| \boldsymbol{\theta}_m^i (t) - \boldsymbol{\theta} (t) \right\|_2^2 \big] + \eta^2 (t) \left( \tau - 1 \right) + 2 \eta(t) (\tau - 1) \Gamma.
\end{align}
We have 
\begin{align}\label{AppLemmaTemr_2_Eq_7}
\frac{1}{M} \sum\limits_{m=1}^{M} \sum\limits_{i=2}^{\tau} \mathbb{E} \big[ \left\| \boldsymbol{\theta}_m^i (t) - \boldsymbol{\theta} (t) \right\|_2^2 \big]  &= \frac{\eta^2(t)}{M} \sum\limits_{m=1}^{M} \sum\limits_{i=2}^{\tau} \mathbb{E} \Big[ \Big\| \sum\nolimits_{l=1}^{i} \nabla F_m \left( \boldsymbol{\theta}_m^l (t), \xi_m^l (t) \right) \Big\|_2^2 \Big] \nonumber\\
& \ainequal \eta^2(t) G^2 \frac{\tau (\tau-1)(2\tau-1)}{6}, 
\end{align}
where (a) follows from the convexity of $\left\| \cdot \right\|_2^2$ and Assumption \ref{AssumpBoundedVarGradient}. 
For $\eta(t) \le 1$, $\forall t$, it follows from \eqref{AppLemmaTemr_2_Eq_6} and \eqref{AppLemmaTemr_2_Eq_7} that
\begin{align}\label{AppLemmaTemr_2_Eq_8}
& \frac{2 \eta(t)}{M} \sum\limits_{m=1}^{M} \mathbb{E} \Big[ \langle {\boldsymbol{\theta}}^* - \boldsymbol{\theta} (t) , \sum\limits_{i=2}^{\tau} \nabla F_m \left( \boldsymbol{\theta}_m^i (t), \xi_m^i (t) \right) \rangle \Big]  \le - \mu \eta(t) (1 - \eta(t)) (\tau - 1) \mathbb{E} \left[ \left\| \boldsymbol{\theta} (t) - \boldsymbol{\theta}^* \right\|_2^2 \right] \nonumber\\
&  \;\;\; +  \left( 1+ \mu (1- \eta(t)) \right) \eta^2(t) G^2 \frac{\tau (\tau-1)(2\tau-1)}{6} + \rho (t) \eta^2 (t) \left( \tau - 1 \right) G^2 + 2  \eta(t) (\tau - 1) \Gamma .
\end{align}
By substituting the results in \eqref{AppLemmaTemr_2_Eq_3} and \eqref{AppLemmaTemr_2_Eq_8} into \eqref{AppLemmaTemr_2_Eq_2}, we obtain
\begin{align}\label{AppLemmaTemr_2_Eq_9}
& 2 \mathbb{E} \left[ \langle \boldsymbol{\theta} (t) - {\boldsymbol{\theta}}^* , \Delta {\boldsymbol{\theta}} (t) \rangle \right]  \le - \mu \eta(t) \left( \tau - \eta(t) (\tau - 1) \right) \mathbb{E} \left[ \left\| \boldsymbol{\theta} (t) - \boldsymbol{\theta}^* \right\|_2^2 \right] \nonumber\\
& \; \;\; +  \left( 1+ \mu (1- \eta(t)) \right) \eta^2(t) G^2 \frac{\tau (\tau-1)(2\tau-1)}{6} + \rho (t) \eta^2 (t) \left( \tau - 1 \right)G^2 + 2  \eta(t) (\tau - 1) \Gamma .
\end{align}
Plugging \eqref{AppLemmaTemr_2_Eq_9} into \eqref{AppLemmaTemr_2_Eq_1_3} completes the proof of Lemma \ref{AppFDPLemmaTerm_2}.

\section{Proof of Lemma \ref{LemmaTermE}}\label{AppProofLemmaTermE}
We have 
\begin{align}\label{EQ_LemmaTermE_app_1}
& \frac{2 \eta(t)}{M} \sum\limits_{m=1}^{M} \sum\limits_{i=2}^{\tau} \mathbb{E} \left[ \langle \boldsymbol{\theta}^* - \boldsymbol{\theta}_m^i (t) , \nabla F_m \left( \boldsymbol{\theta}_m^i (t), \xi_m^i (t) \right) \rangle \right] \nonumber\\
& \qquad \qquad \qquad \ainequal \frac{2 \eta(t)}{M} \sum\limits_{m=1}^{M} \sum\limits_{i=2}^{\tau} \mathbb{E} \left[ \langle \boldsymbol{\theta}^* - \boldsymbol{\theta}_m^i (t) , \nabla F_m \left( \boldsymbol{\theta}_m^i (t) \right) \rangle \right]\nonumber\\
& \qquad \qquad \qquad \thirdinequal \frac{2 \eta(t)}{M} \sum\limits_{m=1}^{M} \sum\limits_{i=2}^{\tau} \mathbb{E} \Big[ F_m (\boldsymbol{\theta}^*) - F_m(\boldsymbol{\theta}_m^i (t)) - \frac{\mu}{2} \left\| \boldsymbol{\theta}_m^i (t) - {\boldsymbol{\theta}}^* \right\|_2^2 \Big]\nonumber\\
& \qquad \qquad \qquad = \frac{2 \eta(t)}{M} \sum\limits_{m=1}^{M} \sum\limits_{i=2}^{\tau} \mathbb{E} \Big[ F_m (\boldsymbol{\theta}^*) - F_m^* + F_m^* - F_m(\boldsymbol{\theta}_m^i (t)) - \frac{\mu}{2} \left\| \boldsymbol{\theta}_m^i (t) - {\boldsymbol{\theta}}^* \right\|_2^2 \Big]\nonumber\\
& \qquad \qquad \qquad = 2 \eta(t) (\tau - 1) \big( F^* - \frac{1}{M} \sum\nolimits_{m=1}^{M} F_m^* \big) +  \frac{2 \eta (t)}{M} \sum\limits_{m=1}^{M} \sum\limits_{i=2}^{\tau} \left( F_m^* - \mathbb{E} \left[ F_m({\boldsymbol{\theta}}_m^i (t)) \right] \right) \nonumber\\
& \qquad \qquad \qquad \quad - \frac{\mu \eta(t)}{M} \sum\limits_{m=1}^{M} \sum\limits_{i=2}^{\tau} \mathbb{E} \left[ \left\| \boldsymbol{\theta}_m^i (t) - {\boldsymbol{\theta}}^* \right\|_2^2 \right] \nonumber\\
& \qquad \qquad \qquad \cinequal 2 \eta(t) (\tau - 1) \Gamma - \frac{\mu \eta(t)}{M} \sum\limits_{m=1}^{M} \sum\limits_{i=2}^{\tau} \mathbb{E} \left[ \left\| \boldsymbol{\theta}_m^i (t) - {\boldsymbol{\theta}}^* \right\|_2^2 \right],  
\end{align}
where (a) follows since $\mathbb{E}_{\xi} \left[ \nabla F_m \left( \boldsymbol{\theta} (t), \xi_m^i (t) \right) \right] = \nabla F_m \left( \boldsymbol{\theta} (t)  \right)$, $\forall i, m, t$, (b) holds because $F_m$ is $\mu$-strongly convex, and (c) follows since $F_m^* \le F_m({\boldsymbol{\theta}}_m^i (t))$, $\forall m, i, t$. 
We have
\begin{align}\label{EQ_LemmaTermE_app_2}
& - \left\| \boldsymbol{\theta}_m^i (t) - {\boldsymbol{\theta}}^* \right\|_2^2 = - \left\| \boldsymbol{\theta}_m^i (t) - {\boldsymbol{\theta}} (t) \right\|_2^2 - \left\| \boldsymbol{\theta} (t) - {\boldsymbol{\theta}}^* \right\|_2^2 - 2 \langle \boldsymbol{\theta}_m^i (t) - {\boldsymbol{\theta}} (t) , \boldsymbol{\theta} (t) - {\boldsymbol{\theta}}^* \rangle \nonumber\\
& \qquad \quad \ainequal - \left\| \boldsymbol{\theta}_m^i (t) - {\boldsymbol{\theta}} (t) \right\|_2^2 - \left\| \boldsymbol{\theta} (t) - {\boldsymbol{\theta}}^* \right\|_2^2 + \frac{1}{\eta(t)} \left\| \boldsymbol{\theta}_m^i (t) - {\boldsymbol{\theta}} (t) \right\|_2^2 + \eta(t) \left\| \boldsymbol{\theta} (t) - {\boldsymbol{\theta}}^* \right\|_2^2 \nonumber\\
&\qquad \quad = - (1 - \eta(t)) \left\| \boldsymbol{\theta} (t) - {\boldsymbol{\theta}}^* \right\|_2^2 + \Big(\frac{1}{\eta(t)} - 1\Big) \left\| \boldsymbol{\theta}_m^i (t) - {\boldsymbol{\theta}} (t) \right\|_2^2,
\end{align}
where (a) follows from the Cauchy-Schwarz inequality. 
The proof of Lemma \ref{LemmaTermE} is completed by substituting the result in \eqref{EQ_LemmaTermE_app_2} into \eqref{EQ_LemmaTermE_app_1}.

\bibliographystyle{IEEEtran}
\bibliography{Report}

\end{document}